\documentclass[lettersize,journal]{IEEEtran}
\usepackage{amsmath,amsfonts}
\usepackage{array}
\usepackage{textcomp}
\usepackage{stfloats}
\usepackage{url}
\usepackage{verbatim}
\usepackage{graphicx}
\usepackage{amsmath,amsthm,amssymb}
\usepackage{mathrsfs}
\usepackage{epstopdf}
\usepackage{authblk}
\usepackage{color}
\usepackage{stmaryrd}
\usepackage{subfigure}
\usepackage{booktabs}
\usepackage{threeparttable}
\usepackage{multirow}
\usepackage{colortbl}
\usepackage{caption}
\usepackage{float}
\usepackage{makecell}
\usepackage{dsfont}
\usepackage{cite}
\usepackage{hyperref}
\usepackage{enumitem}
\usepackage{mathtools}
\usepackage{algorithm}
\usepackage{algpseudocode}

\newtheorem{theorem}{Theorem}

\newtheorem{example}{Example}
\newtheorem{remark}{Remark}
\newtheorem{assumption}{Assumption}
\newtheorem{corollary}{Corollary}

\usepackage{bookmark}

\graphicspath{{figure/}}

\def\BibTeX{{\rm B\kern-.05em{\sc i\kern-.025em b}\kern-.08em
    T\kern-.1667em\lower.7ex\hbox{E}\kern-.125emX}}
\usepackage{balance}
\begin{document}
\title{Distributed Nash Equilibrium Seeking for Open Multi-Coalition Games: A Mass-Preserving Gradient-Tracking-Based Algorithm}
\author{Daning~Lei,~Jiamin~Wang,~Jian~Liu,~Feng~Xiao,~Yuanshi~Zheng
\thanks{This work has been supported in part by the Key R\& D program of Shaanxi Province under Grant 2025CY-YBXM-100.} 
\thanks{D. Lei, J. Wang, J. Liu and Y. Zheng are with School of
Mechano-electronic Engineering, Xidian University, Xi'an,
710071, China. E-mail: leidaning@outlook.com; wangjiamin21@outlook.com; liujianzym@outlook.com; zhengyuanshi2005@163.com.}
\thanks{F. Xiao is with the School of Control and Computer Engineering, North China Electric Power University, Beijing, 102206, China. E-mail: fengxiao@ncepu.edu.cn.}}

\maketitle

\begin{abstract}
	This paper studies distributed information aggregation and Nash equilibrium
	seeking over open multi-coalition networks, where both agents and coalitions
	may dynamically enter and leave the network.
    The varying network population poses a fundamental challenge to conventional gradient-tracking methods, as agent departures may cause the loss of gradient information required for distributed aggregation.
	To overcome this issue, a mass-preserving gradient-tracking method is designed, in which the gradient information carried by departing agents is redistributed to the remaining agents within the same coalition.
	The proposed method preserves the aggregate gradient information of the coalition objective over the intra-coalition communication network and achieves gradient tracking under dynamic population changes without imposing additional initialization conditions beyond those of conventional gradient-tracking methods.
	Based on the proposed gradient-tracking
	mechanism, a distributed Nash equilibrium seeking framework is developed,
	and its convergence is established under both uncoordinated fixed and diminishing
	stepsizes. 
	Furthermore, to alleviate the slowdown of convergence caused by classical diminishing stepsize schemes ($\frac{\eta_0}{t}$, where $\eta_0>0$ is a constant), an exponential-interval diminishing stepsize scheme is designed based on the agents' exploration intervals.
	Finally, a numerical simulation based on an open connecting control game
	demonstrate the effectiveness of the proposed information-tracking and Nash
	equilibrium seeking framework.
\end{abstract}


\section{Introduction}
\IEEEPARstart{D}{istributed} information processing over networked multi-agent
systems has received considerable attention for its ability to accomplish
global computational and control tasks using only local information and
neighbor-to-neighbor communication, with applications in smart grids
\cite{saad2012game,tao2026specified}, mobile sensor networks
\cite{goldenbaum2013harnessing,8048022}, and formation control
\cite{8636259,hu2024time}.
A fundamental challenge in such networked systems is that the information
required to accomplish a global task is distributed among multiple agents.
This challenge is particularly evident in networked decision-making
problems, distributed optimization and
learning, where each agent possesses only a local component of the global
objective and the global information must therefore be aggregated through
inter-agent communication.
The classical gradient-tracking method provides an effective mechanism for such distributed
information aggregation, enabling each agent to asymptotically recover the
global gradient from locally available information and neighbor-to-neighbor
communication \cite{nedic2017achieving}.

Multi-coalition noncooperative games provide a unified framework that
encompasses networked decision-making and distributed optimization, and thus
provide a natural setting for studying distributed information aggregation
\cite{ye2018nash}.
Specifically, agents are
organized into multiple coalitions, where agents within each coalition
cooperate to optimize a common aggregate objective, while different coalitions
compete with each other. 
In such games, agents cooperate within coalitions while
competing with other coalitions, making distributed information processing
essential for obtaining and aggregating the information required for
coalition-level decision making.
A key objective
in the game is to obtain the Nash equilibrium (NE), i.e., a decision profile
at which no coalition can improve its objective by unilaterally changing its
decision. 
Owing to its connection with distributed optimization and networked
decision making, distributed NE seeking for multi-coalition games has been
extensively studied under various information, constraint, and convergence
settings \cite{zhou2022distributed,11119825,meng2023linear,wang2024distributed,
	pang2023distributed,mo2026distributed,wang2025gradient,lei2026distributed}.

In \cite{zhou2022distributed}, a distributed algorithm was developed for
consistency-constrained multi-coalition games by combining consensus and
average gradient tracking. The framework was further extended to
edge-agreement-constrained two-coalition zero-sum games in \cite{11119825}.
Under partial-decision information, Meng et al. \cite{meng2023linear} considered
multi-coalition games with local convex constraints, while
\cite{wang2024distributed,pang2023distributed} employed coalition-decision
estimation and average gradient tracking to handle distributed information
requirements. 
More recently, dynamic decision-making settings have been
considered. 
Mo et al. \cite{mo2026distributed} investigated online
multi-coalition aggregative games with time-varying local constraints using
projected gradient and consensus tracking.
In addition, acceleration techniques such as Nesterov acceleration
and local memory have been exploited to improve the convergence rate of
distributed NE seeking algorithms \cite{wang2025gradient,lei2026distributed}.
Nevertheless, these studies generally assume a fixed set of participating
agents and a fixed population size.

Despite these advances in distributed NE seeking for multi-coalition games,
existing studies generally assume a fixed set of participating agents.
In practical engineering applications, however, the composition of a network
may change over time as agents enter or leave for various reasons.
For example, in power markets, power plants within a company may enter or
leave the market in response to changes in the company's operational
strategies. 
To accommodate dynamic participation, recent studies have
considered open distributed optimization and open noncooperative games
\cite{hsieh2021optimization,11312648,11108040,11157711,11345953,liu2025online,
	liu2026distributed,chen2025adaptive,xu2025dynamic}. 
However, to the best of our
knowledge, distributed NE seeking for open multi-coalition games has not yet
been investigated.

Compared with open $N$-player games and open distributed optimizations, open multi-coalition games involve not
only strategic interactions among coalitions but also cooperation among
agents within each coalition. 
Thus, agents within a coalition need to obtain
the gradient of the coalition objective from information distributed among
multiple coalition members. 
Since this aggregate gradient is not directly
available to any individual agent, existing distributed NE seeking algorithms
for multi-coalition games typically employ average gradient tracking to
aggregate the coalition-wide gradient information
\cite{zhou2022distributed,meng2023linear,wang2024distributed,
	pang2023distributed,mo2026distributed,lei2026distributed}.
However, the
departure of an agent may remove the gradient information carried by its
tracking state, thereby disrupting the information balance required for
accurate gradient tracking.

To address the gradient-tracking challenge in open systems, \cite{11312648}
developed an open distributed optimization algorithm with gradient tracking
that employs acknowledgment messages to detect local topology changes and
max-consensus iterations to disseminate departure information among the
remaining agents. 
In particular, additional max-consensus iterations are
performed between consecutive optimization rounds to identify and propagate
departure events, after which the gradient-tracking states of the remaining
agents are reset accordingly. 
While this mechanism enables gradient tracking
under dynamic agent participation, the additional detection and information
dissemination process introduces extra communication and computation overhead
between optimization updates, which may limit its applicability in scenarios
with frequent agent entries and departures. 

Motivated by the need to maintain gradient tracking under dynamic population
changes without introducing an additional event-detection and dissemination
phase, this paper develops a mass-preserving gradient-tracking mechanism for
open multi-coalition games. 
The main contributions are summarized as follows.
\begin{enumerate}
	\item
	Unlike existing studies, this paper investigates the distributed Nash equilibrium seeking problem for open multi-coalition noncooperative games, where agents are allowed to dynamically enter and leave the game.
	The entry and departure of an agent can only be observed by its neighboring agents. Moreover, since agents that are relevant to an individual agent’s cost function may enter or leave the game, the corresponding cost function may change dynamically, resulting in additional challenges for distributed Nash equilibrium seeking.
	
	\item
	A mass-preserving gradient-tracking mechanism is developed to maintain
	aggregate gradient information under agent departures.
	In particular, the
	gradient information associated with a departing agent is redistributed to
	the remaining agents within the same coalition, thereby preserving the
	aggregate gradient information of the coalition objective over the
	intra-coalition communication network. 
	It is proved that the proposed
	mechanism achieves gradient tracking under dynamic population changes
	without requiring any additional initialization conditions beyond those
	of classical gradient-tracking methods. 
	Based on this mechanism, a
	distributed Nash equilibrium seeking framework is developed, under which
	both constant and diminishing heterogeneous stepsizes are considered.
	For the constant-stepsize case, an explicit upper bound on the
	convergence error is established, whereas sufficient conditions for
	convergence are derived for the diminishing-stepsize case.
	
	\item
	To alleviate the convergence slowdown caused by the continuously
	decreasing stepsize in the classical diminishing scheme
	$\frac{\eta_0}{t}$, where $\eta_0>0$ is a constant, an
	exponential-interval diminishing stepsize scheme is developed based on
	the agents' exploration intervals. 
	Specifically, each agent's time axis
	is partitioned into a countable sequence of exploration intervals with
	increasing lengths, within which the stepsize remains unchanged and is
	reduced only when the interval changes. 
	This design reduces the frequency
	of stepsize variations and mitigates the degradation of convergence speed
	in the later stage of the algorithm.
\end{enumerate}

\indent The following is the organization of this paper: Some preliminaries are presented in Section \uppercase\expandafter{\romannumeral2}. The problem formulation is presented in Section \uppercase\expandafter{\romannumeral3}. Our main results are given in Section \uppercase\expandafter{\romannumeral4}. The numerical simulation is presented in Section \uppercase\expandafter{\romannumeral5} and Section \uppercase\expandafter{\romannumeral6} concludes the paper.


\section{Preliminaries}
\subsection{Notations}
In this paper, $\mathbb{R}$, $\mathbb{R}^{n}$ and $\mathbb{R}^{n\times m}$ denote the real number, the column vector set and the $n\times m$ real matrix set, respectively.
For a set $S$, $| S|$ represents its cardinal number.
The projection onto the closed and convex set $S$ is denoted by $\mathrm{Proj}_{S}(\cdot)$.
Throughout this paper, $\Pi_{i=1}^n S_i$ denotes the product of $S_i$ when $S_i$ is a scalar, whereas it denotes the Cartesian product of $S_i$, $i=1,\ldots,n$, when each $S_i$ is a set.
For a matrix $A\in{\mathbb{R}^{n\times m}}$, let $[A]_i^j$ denote its entry located at the $i$-th row and $j$-th column, respectively. 
Let $\Vert\cdot\Vert_2$, $\lVert\cdot\rVert_F$ and $\left<\cdot,\cdot\right>$ denote the spectral norm, the \emph{Frobenius} norm and the \emph{Frobenius} inner product, respectively.
For any matrices $A,B\in{\mathbb{R}^{n\times m}}$, $\left<A,B\right>=\mathrm{tr}(A^\top B)$.
$A\le B$ means $[A]_i^j\le[B]_i^j$ and $A<B$ means $[A]_i^j<[B]_i^j$.
Let $\mathrm{blk}(A_1,\dots,A_n)$ denote the block diagonal matrix whose $i$-th diagonal block is matrix $A_i$, and $\textrm{col}(x_1,x_2,\cdots,x_n)$ denote the column stack of vectors $x_1,x_2,\cdots,x_n$.
$\mathcal{O}_n$, ${\mathbf 1}_n$, and $I$ denote the $n\times n$ matrix with all entries equal to $0$, the $n$-dimensional column vector with all entries equal to $1$, and the identity matrix of appropriate dimensions, respectively.

Let $\mathcal{G}=\{\mathcal{V},\mathcal{E},\mathcal{A}\}$ denote an undirected graph, where $\mathcal{V}$, $\mathcal{E}\subseteq \mathcal{V}\times \mathcal{V},\mathcal{A}\in{\mathbb{R}^{|\mathcal{V}|\times |\mathcal{V}|}}$ are the vertex set, the edge set and the adjacency matrix, respectively.
Let $\mathrm{col}\{x_{i}\}_{v_{i}\in{\mathcal{V}}}$ denote the stack column vector of column vectors $\{x_i\}_{v_i\in{\mathcal{V}}}$.
Let $\mathrm{blk}\{A_i\}_{v_i\in{\mathcal{V}}}$ denote the block diagonal matrix whose diagonal blocks consist of  $A_i$.
\section{Problem formulation}
	This paper considers the open multi-coalition game, where agents are allowed to dynamically join or leave the game.
Agents within the same coalition
cooperate to minimize the coalition's aggregate cost, while different
coalitions compete with each other.

\subsection{Open Multi-Coalition Game Model}
Let $\mathcal{C}_t=\{1,\dots,N_t\}$ denote the set of coalitions participating
in the game at time $t$, where $N_t$ is the number of participating
coalitions. The set of agents participating in coalition $i$ is denoted by
$\mathcal{V}_i^t$, and the overall agent set is given by
$
\mathcal{V}^t
=
\bigcup_{i\in\mathcal{C}_t}\mathcal{V}_i^t.
$
Let $n_{i,t}=|\mathcal{V}_i^t|$ denote the number of agents in coalition $i$
at time $t$. 
Accordingly, the total number of participating agents is
$
\widetilde n_t
=
\sum_{i=1}^{N_t} n_{i,t}.
$
The remaining and departing agents of coalition $i$ at time $t$ are,
respectively, denoted by
$
\mathcal{R}_i^t
=
\mathcal{V}_i^t\cap\mathcal{V}_i^{t-1},
\mathcal{D}_i^t
=
\mathcal{V}_i^{t-1}\setminus\mathcal{V}_i^t.
$
The entry and departure of agents satisfy the following rules:
\begin{enumerate}
	\item \textbf{Agent entry:}
	Each entering agent is regarded as a new agent and assigned a unique
	identifier.
	
	\item \textbf{Agent departure:}
	No departing agent leaves simultaneously with all of its intra-coalition
	neighbors.
\end{enumerate}
\begin{remark}
	In this paper,  we consider a class of dynamic systems in which agent entry is unrestricted, whereas agent departure is subject to certain constraints. Such a system structure helps prevent the loss of critical information when agents leave the system. This is analogous to a company in which, rather than allowing all employees of a department to leave simultaneously, some experienced employees are retained to ensure the continuity of the department's operations.
	It is worth noting that the network structure considered in \cite{liu2026distributed}, which has a fixed set of nodes but a dynamically changing agent composition, provides a representative example of such a system.
\end{remark}

Without loss of generality, let $t_\iota$, $\iota=1,2,\dots$, denote the
$\iota$-th switching instant at which the membership of the game changes due
to the entry or departure of an agent. 
The time axis is therefore divided
into consecutive intervals
$
[t_\iota,t_{\iota+1}).
$
For any
$t\in[t_\iota,t_{\iota+1})$, no agent enters or leaves the game. Hence,
$\mathcal{C}_t$ and $\mathcal{V}_i^t$ remain unchanged within this interval.
Consequently, the feasible
sets and cost functions associated with the current game configuration also
remain unchanged during $[t_\iota,t_{\iota+1})$.

For $t_\iota\le t<t_{\iota+1}$, let
$x_{i,j}(t)\in\Omega_{i,j}\subseteq\mathbb{R}$ denote the state of agent
$v_{i,j}$, where $\Omega_{i,j}$ is a nonempty, compact, and convex set. The
collective state of coalition $i$ is given by
$
x_i(t)
=
\operatorname{col}\{x_{i,j}(t)\}_{v_{i,j}\in\mathcal{V}_i^{t_\iota}}
\in\Omega_i^{t_\iota},
$
where
$
\Omega_i^{t_\iota}
=
\prod_{j=1}^{n_{i,t_\iota}}\Omega_{i,j}.
$
Let
$
x_{-i}(t)
=
\operatorname{col}\{x_k(t)\}_{k\in
	\mathcal{C}_{t_\iota}\setminus\{i\}}
$
denote the collective state of all coalitions except coalition $i$, and let
$
x(t)
=
\operatorname{col}\{x_k(t)\}_{k\in\mathcal{C}_{t_\iota}}
\in\Omega^{t_\iota},
$
where
$
\Omega^{t_\iota}
=
\prod_{i\in\mathcal{C}_{t_\iota}}\Omega_i^{t_\iota}.
$

In general, the cost of each agent may depend on the states of multiple
agents. To characterize such state-dependent interactions, let
$\mathcal{U}_{i,j}^{t_\iota}$ denote the set of payoff-relevant agents
associated with $v_{i,j}$, namely, the agents whose states directly affect
the cost of $v_{i,j}$ during $[t_\iota,t_{\iota+1})$ (see Example 1). 
Accordingly, the local
cost function of agent $v_{i,j}$ is denoted by
$f_{i,j}^{t_\iota}(\mathbf{x}_{i,j})$, where
$
\mathbf{x}_{i,j}
=
\operatorname{col}
\{x_{h,l}\}_{v_{h,l}\in\mathcal{U}_{i,j}^{t_\iota}}.
$
The aggregate cost of coalition $i$ is defined as
$
f_i^{t_\iota}(x_i,x_{-i})
=
\sum_{v_{i,j}\in\mathcal{V}_i^{t_\iota}}
f_{i,j}^{t_\iota}(\mathbf{x}_{i,j}).
$
The objective of coalition $i$ is therefore
\begin{equation}
	\min_{x_i\in\Omega_i^{t_\iota}}
	\quad
	f_i^{t_\iota}(x_i,x_{-i}).
	\label{106-1}
\end{equation}
\begin{example}
	For agent $v_{1,1}$, if its cost function is $f_{1,1}^{t_1}(x_{1,1},x_{1,3},x_{2,4})=x_{1,1}^2+x_{1,3}x_{2,4}$ at $t_1$, then $\mathcal{U}_{1,1}^{t_1}=\{v_{1,1},v_{1,3},v_{2,4}\}$.
\end{example}

To solve \eqref{106-1}, each agent needs both the state
information of the payoff-relevant agents and the gradient information of
its coalition objective. Accordingly, agent $v_{i,j}$ maintains two
estimators
\[
y_{i,j}(t)
=
\operatorname{col}
\{y_{i,j}^{h,l}(t)\}_{v_{h,l}\in\mathcal{U}_{i,j}^{t_\iota}}
\]
and
\[
s_{i,j}(t)
=
\operatorname{col}
\{s_{i,j}^{l}(t)\}_{v_{i,l}\in
	(\mathcal{N}_{i}^{t_\iota}\cup\mathcal{V}_i^{t_\iota})
	\cap\mathcal{U}_{i,j}^{t_\iota}},
\]
where $y_{i,j}(t)$ estimates the states contained in
$\mathbf{x}_{i,j}$ and $s_{i,j}(t)$ estimates the gradient information
required for solving \eqref{106-1}. In particular,
$s_{i,j}(t)$ is used to track the coalition-level gradient
$
\nabla_i f_i^{t_\iota}(x_i,x_{-i}).
$

For a given switching instant $t_\iota$, the strategy profile
$
x_{t_\iota}^*
=
(x_{i,t_\iota}^*,x_{-i,t_\iota}^*)
\in\Omega^{t_\iota}
$
is called a Nash equilibrium (NE) if
\[
f_i^{t_\iota}
(x_{i,t_\iota}^*,x_{-i,t_\iota}^*)
\le
f_i^{t_\iota}
(x_i,x_{-i,t_\iota}^*),
\qquad
\forall x_i\in\Omega_i^{t_\iota},
\quad
\forall i\in\mathcal{C}_{t_\iota}.
\]
The objective of the proposed algorithm is to enable each participating
agent to obtain the current NE through distributed information exchange.
\subsection{Communication Mechanism}
To support the distributed information processing required by the
multi-coalition game, two communication networks are considered. 
As illustrated in Fig. 1, a schematic representation of the network structure is provided.
\begin{figure}
	\centering
	\includegraphics[height=0.6\linewidth,width=1\linewidth]{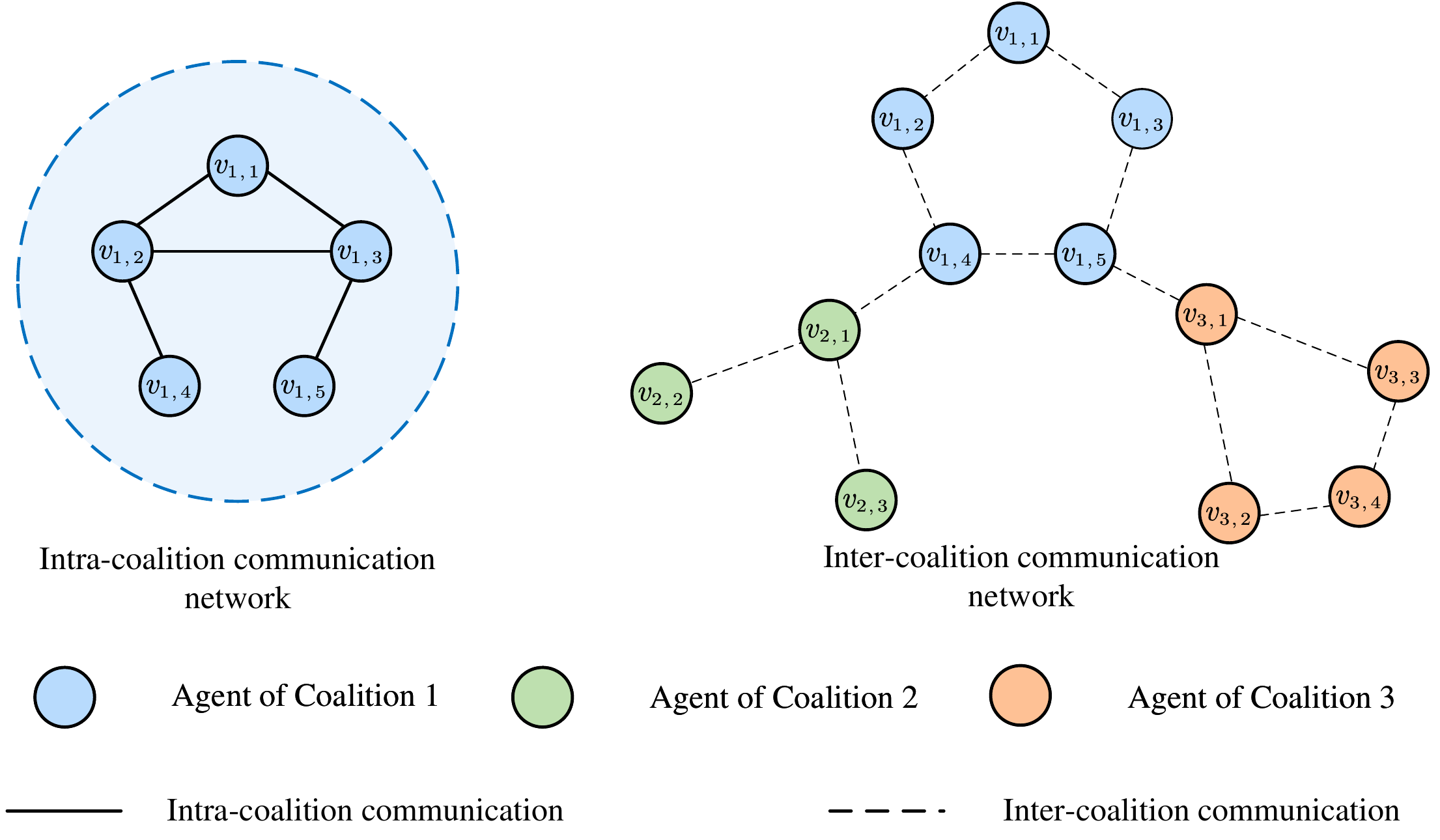}
	\captionsetup{font=footnotesize}
	\caption{An illustration of the intra-networks and inter-network in a $3$-coalition game.}
\end{figure}
The
intra-coalition communication network of coalition $i$ is described by the
undirected graph
$
\mathcal{G}_i^{t_\iota}
=
(\mathcal{V}_i^{t_\iota},
\mathcal{E}_i^{t_\iota},
\mathcal{A}_i^{t_\iota}),
$
which consists exclusively of agents belonging to coalition $i$. The
inter-coalition communication network is described by
$
\mathcal{G}^{t_\iota}
=
(\mathcal{V}^{t_\iota},
\mathcal{E}^{t_\iota},
\mathcal{A}^{t_\iota}),
$
which characterizes the communication relationships among all participating
agents. The corresponding neighbor sets are defined as
\[
\mathcal{N}_{v_{i,j}}^{t_\iota}
=
\left\{
v_{h,l}\mid
(v_{i,j},v_{h,l})\in\mathcal{E}^{t_\iota}
\right\}
\]
and
\[
\mathcal{N}_{v_{i,j}}^{i,t_\iota}
=
\left\{
v_{i,l}\mid
(v_{i,j},v_{i,l})\in\mathcal{E}_{i}^{t_\iota}
\right\}.
\]
The intra-coalition network is used to aggregate coalition-level gradient
information, whereas the inter-coalition network enables the exchange of
state information required to evaluate coupled cost functions.

During
$[t_\iota,t_{\iota+1})$, the corresponding communication networks remain
unchanged. 
When an agent leaves the game, its gradient-tracking information
is no longer directly available to the remaining agents. 
To preserve the
aggregate gradient information of the coalition, each departing agent
$v_{i,j}$ transfers the information associated with its loss mass to its
remaining intra-coalition neighbors at the switching instant. Denote the
loss mass associated with the departure of $v_{i,j}$ by
$Mass_{i,j}$.
This mass is redistributed among the remaining
intra-coalition neighbors and incorporated into their gradient-tracking
states. 
Consequently, the aggregate gradient information carried by the
intra-coalition network can be preserved despite the change in coalition
membership.

\section{Distributed algorithm}
In this section, we present the mass-preserving gradient-tracking based Nash seeking algorithm for the open multi-coalition games.
The proposed algorithm consists of three layers: an initialization layer, a state-update layer, and a gradient-tracking layer. 
The initialization layer is activated when an agent joins the game and is responsible for initializing the state and gradient estimators.
The state-update layer is responsible for driving the agents' states toward the NE, while the gradient-estimation layer is designed to estimate and track the gradient of the aggregate cost function associated with each coalition.

Next, we present the overall structure and main components of the proposed algorithm.

Let $
q_{i,j}^{t_\iota}(t)=\nabla_i f_{i,j}^{t_\iota}(y_{i,j}(t)),
$
and denote its $m$-th component by $q_{i,j}^{m,t_\iota}(t)=[q_{i,j}^{t_\iota}(t)]_m.$
Without loss of generality, suppose that the agents in $\mathcal{V}^{t_\iota}$ are indexed according to a certain predetermined order. 
Let $\tau_{i,j}$ denote the index of agent $v_{i,j}$ under this ordering.

\textbf{Initialization layer:} For all $v_{i,j}$, let the initialization time $T_{i,j}$ represent the moment when $v_{i,j}$ enters the game.
At $T_{i,j}$, arbitrarily select a value from $\mathbb{R}^{|\mathcal{U}_{i,j}^{T_{i,j}}\cup\mathcal{N}_{i,j}^{T_{i,j}}|}$ as the initial state estimator of $v_{i,j}$, i.e., $y_{i,j}(T_{i,j})\in{\mathbb{R}^{|\mathcal{U}_{i,j}^{T_{i,j}}\cup\mathcal{N}_{i,j}^{T_{i,j}}|}}$. 
Furthermore, let $s_{i,j}(T_{i,j})=q_{i,j}^{T_{i,j}}(T_{i,j})$ and $q_{i,j}^{m,T_{i,j}-1}(T_{i,j})=q_{i,j}^{m,T_{i,j}}(T_{i,j})$.

\textbf{State update layer:}
Without loss of generality, assume that agent $v_{i,j}$ has access to at least one estimate $y_{h,l}^{k,m}$ on the state of agent $v_{k,m}$. 
The relationship between $v_{k,m}$ and $v_{i,j}$ can be classified into the following three cases:
$$
\text{(i) } k=i,\ m=j;
\qquad
\text{(ii) } k=i,\ m\neq j;
\qquad
\text{(iii) } k\neq i.
$$
Specifically, these three cases correspond to the situation where $v_{k,m}$ is $v_{i,j}$ itself, another agent within the same coalition, and an agent belonging to a different coalition, respectively.
Then, the state estimate $y_{i,j}^{k,m}$ is updated according to the following rules for $t\in[t_\iota,t_{\iota+1}-1]$.

\noindent\textbf{Case 1:} 
$v_{k,m}\in\mathcal{N}_{i,j}^{t_\iota}\cup\mathcal{U}_{i,j}^{t_\iota}$,
or $v_{i,j}$ maintains the state estimate $y_{i,j}^{k,m}(t)$.

If $k=i$ and $m=j$.
\begin{flalign}\label{9-1}
	&y_{i,j}^{i,j}(t+1)&\notag\\
	&=\mathrm{Proj}_{\Omega_{i,j}}\Bigg(
	y_{i,j}^{i,j}(t)
	+\varepsilon
	\sum_{v_{h,l}\in\mathcal{N}_{i,j}^{t_\iota}}
	[\mathcal{A}]_{\tau_{i,j}}^{\tau_{h,l}}
	\big(y_{h,l}^{i,j}(t)-y_{i,j}^{i,j}(t)\big)\notag\\
	&-\eta_{i,j}s_{i,j}^{j}(t)
	\Bigg).
\end{flalign}

If $k=i$ and $m\neq j$.
\begin{flalign}\label{103-1}
	&y_{i,j}^{i,m}(t+1)\notag\\
	&=y_{i,j}^{i,m}(t)
	+\varepsilon
	\sum_{v_{h,l}\in\mathcal{N}_{i,j}^{t_\iota}}
	[\mathcal{A}]_{\tau_{i,j}}^{\tau_{h,l}}
	\big(y_{h,l}^{i,m}(t)-y_{i,j}^{i,m}(t)\big)
	-\eta_{i,j}s_{i,j}^{m}(t).
\end{flalign}

If $k\neq i$.
\begin{flalign}\label{9-2}
	&y_{i,j}^{k,m}(t+1)=y_{i,j}^{k,m}(t)
	+\varepsilon
	\sum_{v_{h,l}\in\mathcal{N}_{i,j}^{t_\iota}}
	[\mathcal{A}]_{\tau_{i,j}}^{\tau_{h,l}}
	\big(y_{h,l}^{k,m}(t)-y_{i,j}^{k,m}(t)\big).
\end{flalign}

\noindent\textbf{Case 2:}
$v_{k,m}\notin\mathcal{N}_{i,j}^{t_\iota}\cup\mathcal{U}_{i,j}^{t_\iota}$,
and $v_{i,j}$ does not maintain the state estimate $y_{i,j}^{k,m}(t)$.

If $k=i$
\begin{flalign}\label{105-1}
	&y_{i,j}^{k,m}(t+1)=y_{i,j}^{k,m}(t)
	+\varepsilon
	\sum_{v_{h,l}\in\mathcal{N}_{i,j}^{t_\iota}}
	[\mathcal{A}]_{\tau_{i,j}}^{\tau_{h,l}}
	y_{h,l}^{k,m}(t)-\eta_{i,j}s_{i,j}^m(t).&
\end{flalign}
\indent If $k\neq i$
\begin{flalign}\label{1-2}
	&y_{i,j}^{k,m}(t+1)=y_{i,j}^{k,m}(t)
	+\varepsilon
	\sum_{v_{h,l}\in\mathcal{N}_{i,j}^{t_\iota}}
	[\mathcal{A}]_{\tau_{i,j}}^{\tau_{h,l}}
	y_{h,l}^{k,m}(t).
\end{flalign}
It is important to note that in the state update layer, if agent $v_{h,l}$ does not maintain $y_{h,l}^{k,m}(t)$, then $y_{h,l}^{k,m}(t)$ is set to zero when performing the corresponding computations.



\begin{remark}
	It can be observed from the state-estimation update rules that the dimension of the state estimate maintained by each agent varies dynamically over time. Through information propagation over the communication network, each agent can progressively establish estimates of the states of all agents in the network.
\end{remark}

Let $Mass_{i,j}=\frac{1}{|\mathcal{N}_{i,j}^{i,t_{\iota-1}}\cap\mathcal{V}_{i}^{t_\iota}|}\big(s_{i,j}(t_\iota)-q_{i,j}^{t_{\iota-1}}(t_\iota)\big)$ for $v_{i,j}\in{\mathcal{V}_{i}^{t_{\iota-1}}}$.  




\textbf{Gradient estimation layer:} 
Before performing the iterative updates of the algorithm, the following operation is first executed.
Specifically, if $v_{i,j}$ receives $Mass_{i,l}$, then
\begin{flalign}\label{2-1}
	s_{i,j}(t_\iota)\leftarrow s_{i,j}(t_\iota)+\sum\limits_{v_{i,l}\in{\mathcal{N}_{i,j}^{i,t_{\iota-1}}}\cap\mathcal{D}_{t_\iota}}Mass_{i,l},
\end{flalign}
where ``$\leftarrow$'' represents the assignment operation.
Then, for $t_\iota\le t\le t_{\iota+1}-1$, the gradient estimate of $v_{i,j}$ is updated according to the following rules.

 If $v_{k,m}\in{(\mathcal{N}_{i,j}^{i,t_{\iota}}\cup\mathcal{V}_i^{t_{\iota}}\cap\mathcal{U}_{i,j}^{t_{\iota}})}$, or $v_{i,j}$ already maintains the corresponding gradient estimate $s_{i,j}^m(t)$, 
\begin{equation}
	s_{i,j}^m(t+1) = 
			\begin{cases}
			s_{i,j}^m(t)+\varepsilon\sum\limits_{v_{i,l}\in{\mathcal{N}_{i,j}^{i,t_\iota}}}[\mathcal{A}^i]_j^l(s_{i,l}^m(t)-s_{i,j}^m(t))\\
			+q_{i,j}^{m,t_\iota}(t+1)-q_{i,j}^{m,t_{\iota-1}}(t), \quad t=t_\iota, \\
			s_{i,j}^m(t)+\varepsilon\sum\limits_{v_{i,l}\in{\mathcal{N}_{i,j}^{i,t_{\iota}}}}[\mathcal{A}^i]_j^l(s_{i,l}^m(t)-s_{i,j}^m(t))\\
			+q_{i,j}^{m,t_{\iota}}(t+1)-q_{i,j}^{m,t_{\iota}}(t), \quad t_\iota<t\le t_{\iota+1}-1. 
		\end{cases}
	\label{5-1}
\end{equation}

 If $v_{k,m}\notin{(\mathcal{N}_{i,j}^{i,t_{\iota}}\cup\mathcal{V}_i^{t_{\iota}}\cap\mathcal{U}_{i,j}^{t_{\iota}})}$, and $v_{i,j}$ does not maintain $s_{i,j}^{m}(t)$,
\begin{flalign}\label{1-3}
	&s_{i,j}^m(t+1)=\varepsilon\sum\limits_{v_{i,l}\in{\mathcal{N}_{i,j}^{i,t_{\iota}}}}[\mathcal{A}_{t_\iota}^i]_j^ls_{i,l}^m(t).
\end{flalign}

Similarly, in the gradient estimation layer, if agent $v_{i,l}$ does not maintain $s_{i,l}^{m}(t)$, then $s_{i,l}^{m}(t)$ is set to zero when performing the corresponding computations.

For comparison, the open gradient-tracking method in \cite{11312648}
employs acknowledgment messages to detect topology changes and a distributed
max-consensus procedure to disseminate departure information among the
remaining agents. 
A departure event subsequently triggers a coordinated reset
of the gradient-tracking states.
The proposed mass-preserving
mechanism preserves the gradient information associated with a departing agent
by transferring its information to the remaining intra-coalition neighbors,
without requiring an additional event-detection and information-dissemination
procedure between consecutive optimization updates. 
This provides a more
direct information-preserving mechanism for gradient tracking under dynamic
agent participation.

The pseudo code of the proposed algorithm is given in \textbf{Algorithm 1}.
\begin{algorithm}[!htp]
	\caption{Mass-preserving gradient-tracking-based Nash equilibrium seeking algorithm}
	\begin{algorithmic}[1]
		\State \textbf{Convention:} An unmaintained state estimates are treated as zero.
		\Statex \textbf{Initialization layer:}
		
		\For{each agent $v_{i,j}$}
		\State Set $T_{i,j}$ as the initialization time when $v_{i,j}$ enters the game.
		\State Initialize $y_{i,j}(T_{i,j})$ arbitrarily.
		\State Initialize $s_{i,j}(T_{i,j})$ and $q_{i,j}^{m,T_{i,j}-1}(T_{i,j})$ as specified above.
		\EndFor
		
\Statex \textbf{State update layer:}
\For{any $t_\iota$}
\For{$t=t_\iota,\ldots,t_{\iota+1}-1$}
\For{each pair $(i,j)$ and $(k,m)$}
\If{$v_{k,m}\in\mathcal{N}_{i,j}^{t_\iota}\cup\mathcal{U}_{i,j}^{t_\iota}$
	\textbf{or} $y_{i,j}^{k,m}(t)$ is maintained}

\If{$k=i$ and $m=j$}
\State Update $y_{i,j}^{k,m}(t+1)$ according to \eqref{9-1}

\ElsIf{$k=i$ and $m\neq j$}
\State Update $y_{i,j}^{k,m}(t+1)$ according to \eqref{103-1}

\Else
\State Update $y_{i,j}^{k,m}(t+1)$ according to \eqref{9-2}
\EndIf

\ElsIf {$k=i$}
\State Update $y_{i,j}^{k,m}(t+1)$ according to \eqref{105-1}

\Else {$k\neq i$}
\State Update $y_{i,j}^{k,m}(t+1)$ according to \eqref{1-2}

\EndIf
\EndFor
\EndFor
\EndFor

		\Statex \textbf{Gradient estimation layer:}
		\If{agent $v_{i,j}$ receives the quantity $Mass_{i,l}$}
		\State Update $s_{i,j}(t_\iota)$ according to \eqref{2-1}.
		\EndIf
		\Statex
		\For {any $t_\iota$} 
		\For{$t=t_\iota,\ldots,t_{\iota+1}-1$}
		\For{each relevant agent $v_{k,m}$}
		\If{$v_{k,m}\in
			\mathcal{N}_{i,j}^{i,t_\iota}
			\cup
			(\mathcal{V}_i^{t_\iota}\cap\mathcal{U}_{i,j}^{t_\iota})$
			\textbf{or} $v_{i,j}$ maintains $s_{i,j}^{m}(t)$}
		\If{$t=t_\iota$}
		\State Update $s_{i,j}^{m}(t+1)$ according to the first case of \eqref{5-1}.
		\Else
		\State Update $s_{i,j}^{m}(t+1)$ according to the second case of \eqref{5-1}.
		\EndIf
		\Else
		\State Update $s_{i,j}^{m}(t+1)$ according to \eqref{1-3}.
		\EndIf
		\EndFor
		\EndFor
		\EndFor
		
	\end{algorithmic}
\end{algorithm}

\begin{remark}
It can be observed from the above algorithm that an agent may continue to maintain and transmit estimates associated with agents that have already left the game. 
Although this may incur additional computational and memory overhead, it does not affect the NE seeking performance of the proposed algorithm. 
In this work, we focus primarily on the NE seeking performance of the proposed algorithm. 
Since retaining the estimates associated with agents that have left the game does not affect such performance, we do not introduce an additional mechanism to remove these estimates. 
But, the removal of outdated estimates will be considered in future work.
\end{remark}
\section{Convergence analysis}
This section establishes the convergence properties of the proposed distributed
Nash equilibrium seeking algorithm under dynamic agent participation. 
The
analysis is developed in two stages. 
First, we show that the proposed
mass-preserving mechanism maintains the aggregate gradient information of each
coalition despite agent entries and departures. 
Next, based on
this result, we analyze the convergence of the distributed Nash equilibrium
seeking dynamics under both uncoordinated fixed stepsizes and diminishing stepsizes.

For the subsequent analysis, we impose the following conditions. 
\begin{assumption}
	For a given constant $t_\iota$, the functions and its gradients satisfy the following conditions, respectively.
	\begin{enumerate}[]
		\item $f_{i,j}^{t_\iota}(\mathbf{x}_{i,j})$ is continuously differentiable;
		\item $\Vert \nabla_i f_{i,j}^{t_\iota}(\mathbf{x}_{i,j})-\nabla_i f_{i,j}^{t_\iota}(\mathbf{y}_{i,j})\Vert_F\le \beta_{i,j,{t_\iota}}\Vert \mathbf{x}_{i,j}-\mathbf{y}_{i,j}\Vert_F$, for any $\mathbf{x}_{i,j},\mathbf{y}_{i,j}\in{\boldsymbol{\Pi}_{v_{i,j}\in{\mathcal{U}_{i,j}^{t_\iota}}}\Omega_{i,j}}$, where $\beta_{i,j,{t_\iota}}$ is a positive constant;
		\item $\left<g^{t_\iota}(x)-g^{t_\iota}(y),x-y\right>\ge\alpha_{t_\iota}\Vert x-y\Vert_F^2$, for $x,y\in{\Omega^{t_\iota}}$, where $\alpha_{t_\iota}$ is a positive constant, $g^{t_\iota}(x)=\mathrm{col}\{g_i^{t_\iota}(\mathbf{x}_{i,j})\}^\top_{i\in{\mathcal{C}_{t_\iota}}}$ and $g_i^{t_\iota}(\mathbf{x}_{i,j})=\nabla_if_{i}^{t_\iota}(\mathbf{x}_{i,j})$.
	\end{enumerate}
\end{assumption}
\begin{assumption}
	$\mathcal{G}_{t_\iota},\mathcal{G}_{1,t_\iota},\dots, \mathcal{G}_{N_{t_\iota},t_\iota}$ are connected and without self-loops, for $\iota=1,2,\dots$.
\end{assumption}
\begin{assumption}
	There is a positive constant $n_{\text{tal}}$ such that $\widetilde{n}_{t_\iota}\le n_{\text{tal}}$, for all $t_\iota$.
\end{assumption}
\begin{remark}
	In real-world applications, a platform or environment can accommodate only a finite number of agents at any given time due to practical limitations such as limited physical space, communication channels, and computational or communication resources. 
	Assumption 3 captures this practical constraint by imposing an upper bound on the number of agents that can simultaneously participate in the game.	
\end{remark}

We first establish the following theorem to characterize the effectiveness of the proposed gradient-tracking mechanism.
\begin{theorem}
	If $s_{i,j}(T_{i,j})=q^{t_\iota}_{i,j}(T_{i,j})$ for arbitrary $v_{i,j}$, then $\sum_{v_{i,j}\in{\mathcal{V}_i^{t_\iota}}}s_{i,j}(t)=\sum_{v_{i,j}\in{\mathcal{V}_i^{t_\iota}}}q_{i,j}^{t_\iota}(t)$ for $\iota=1,2\dots,$ and $t_\iota<t\le t_\iota-1$.
\end{theorem}
\begin{proof}
	\indent The proof is presented in Appendix.
\end{proof}

\begin{remark}
	In a static multi-coalition game, according to the standard gradient-tracking algorithms in \cite{11119825,liu2026distributed1,yu2023distributed,liu2024distributed,11157711}, if the gradient estimators are initialized as $s_{i,j}(0)=q_{i,j}(0)$, then it follows that
	$$
	\sum_{v_{i,j}\in\mathcal{V}_i}s_{i,j}(t)
	=
	\sum_{v_{i,j}\in\mathcal{V}_i}q_{i,j}(t),
	\qquad \forall t.
	$$
	This property ensures that the gradient estimator maintained by each agent can track the gradient of the aggregate cost function of its corresponding coalition. 
	However, in an open system, the departure of agents may disrupt this equality, since the gradient information associated with a departing agent may no longer be preserved in the intra-networks.
\end{remark}
According to Theorem 1, if the gradient estimator of each entering agent
$v_{i,j}$ is initialized at its entry time $T_{i,j}$ according to
$
s_{i,j}(T_{i,j})=q_{i,j}^{t_\iota}(T_{i,j}),
$ then the proposed gradient-tracking mechanism preserves the aggregate gradient
information of each coalition and achieves gradient tracking in the considered
open system. 
Notably, the initialization conditions required by Theorem 1 are
essentially consistent with those of classical gradient-tracking methods,
despite the dynamic changes in the participating agent set.
Moreover, when the number of coalitions in the open multi-coalition game is
fixed at one, the considered problem reduces to an open distributed
optimization problem. Therefore, the proposed mass-preserving
gradient-tracking mechanism can be directly applied to distributed
optimization problems with dynamically changing agent populations.

Let $Y_i(t)=\mathrm{col}\{y^\top_{i,j}(t)\}_{v_{i,j}\in{\mathcal{V}_i^{t_\iota}}}\in{\mathbb{R}^{n_{i,t_\iota}\times \widetilde{n}_{t_\iota}}}$ and $Y(t)=\mathrm{col}\{Y_i(t)\}_{i\in{\mathcal{C}_{t_\iota}}}\in{\mathbb{R}^{\widetilde{n}_{t_\iota}\times \widetilde{n}_{t_\iota}}}$.
Let $S_i(t)=\mathrm{col}\{s^\top_{i,j}(t)\}_{v_{i,j}\in{\mathcal{V}_i^{t_\iota}}}\in{\mathbb{R}^{n_{i,t_\iota}\times n_{i,t_\iota}}}$ and $S(t)=\mathrm{blk}\{S_i(t)\}_{i\in{\mathcal{C}_{t_\iota}}}\in{\mathbb{R}^{\widetilde{n}_{t_\iota}\times \widetilde{n}_{t_\iota}}}$.
Let $Q_{i,t_\iota}(t)=\mathrm{col}\{(q^{t_\iota}_{i,j}(t))^\top\}_{v_{i,j}\in{\mathcal{V}_i^{t_\iota}}}\in{\mathbb{R}^{n_{i,t_\iota}\times n_{i,t_\iota}}}$ and $Q_{t_\iota}(t)=\mathrm{blk}\{Q_{i,t_\iota}(t)\}_{i\in{\mathcal{C}_{t_\iota}}}\in{\mathbb{R}^{\widetilde{n}_{t_\iota}\times \widetilde{n}_{t_\iota}}}$.
Let $\kappa_{t_\iota}=t_{\iota+1}-t_\iota$ denote the total number of time points between $t_{\iota+1}$ and $t_\iota$.
Denote $\eta_{\text{max}}=max_{v_{i,j}\in{\mathcal{V}^{t_\iota}}}\{\eta_{i,j}\}$, and 
$\eta_{\text{min}}=\min_{v_{i,j}\in{\mathcal{V}_i^{t_\iota}}}\{\eta_{i,j}\}$.
Furthermore, denote $\eta_{\text{inf}}=\min_{t_\iota}\min_{v_{i,j}\in{\mathcal{V}_i^{t_\iota}}}\{\eta_{i,j}\}$ and $\eta_{\text{sup}}=\max_{t_\iota}\max_{v_{i,j}\in{\mathcal{V}_i^{t_\iota}}}\{\eta_{i,j}\}$.
Then, if Assumption 3 holds, $\eta_{\text{inf}}$ and $\eta_{\text{sup}}$ are fixed constants.
	Moreover, since $\Omega_{i,j}$ is compact, there exists a constant $B_x>0$ such that
$
\Vert x_{i,j}(t)\Vert_F\le B_x,
$
and consequently,
$
\Vert x(t)\Vert_F\le \sqrt{\widetilde{n}_{t\iota}}B_x,$ for $t_\iota\le t<t_{\iota+1}.$
Furthermore, if Assumption 1 holds, then it follows that
$
\Vert Q_{t_\iota}\Vert_F\le 2\beta_{t_\iota} \widetilde{n}_{t_\iota}B_x$ for $t_\iota\le t<t_{\iota+1},$
where $\beta_{t_\iota}=\max_{i\in{\mathcal{C}_{t_\iota}}}\{\beta_{i,t_\iota}\}$ and $\beta_{i,t_\iota}=\max_{v_{i,j}\in{\mathcal{V}_i^{t_\iota}}}\{\beta_{i,j,t_\iota}\}.$

Define $e(t)\coloneq(
	\Vert Y(t)-\mathbf{1}_{\widetilde{n}_{t_\iota}}x_{t_\iota}^*\Vert_F	,
	\Vert S(t)-\bar{S}(t)\Vert_F)^\top$, where $\bar{S}=\mathrm{blk}\{\frac{\mathbf{1}_{n_{i,t_{iota}}}\mathbf{1}_{n_{i,t_{iota}}}^{\top}\bar{S}_i(t)}{n_{i,t_{iota}}}\}_{i\in{\mathcal{C}_{t_\iota}}}$.
Then, we can obtain the following theorem.

\begin{theorem}
	Suppose Assumption 1 holds.
	If $\kappa_{t_\iota}>1$ and $\varepsilon<\frac{1}{n_{\text{tal}}}$, then, 
	\begin{align*}
		\Vert e(t_{\iota+1})\Vert_F\le\Vert \hat{\Gamma}_{t_\iota+1}^{\kappa_{t_\iota}-1}\Vert_F(\Vert e(t_\iota)\Vert_F+4\beta_{t_\iota} \widetilde{n}_{t_\iota}B_x),
	\end{align*}
	where $\hat{\Gamma}_{t_\iota}$ is a matrix defined as $\eqref{8-1}$.
	Furthermore, If Assumptions 2 and 3 also hold, $\varepsilon<\frac{1}{n_{\text{tal}}}$ and $\eta_{\mathrm{max}}$ satisfies \eqref{6-10}, then $\forall t_\iota$, $\rho(\hat{\Gamma}_{t_\iota})<1$, that is, the algorithm exhibits exponential convergence.

	  \begin{align}\label{6-10}
		\eta_{\text{max}}<\min&\left\{\frac{\sqrt{n_{tal}}}{2\alpha_{t_\iota}},\frac{\eta_{\text{inf}}\alpha_{t_\iota}(1-\sigma_{t_\iota})}{\beta_{t_\iota}^2\eta_{\text{sup}}\sqrt{\widetilde{n}_{t_\iota}}},\sqrt{\frac{1-\sigma_{t_\iota}^2}{2\beta_{t_\iota}^2}},\right.&\notag\\
		&\quad \frac{(1-\hat{\sigma}_{t_\iota})(1-\sqrt{\rho(\Gamma_{t_\iota})})}{4\beta_{t_\iota}},\notag\\ &\quad\left.\frac{(1-\sigma_{t_\iota}^2)\eta_{\text{inf}}\alpha_{t_\iota}}{2\beta_{t_\iota}\big((1+\sigma_{t_\iota})\beta_{t_\iota}\eta_{\text{sup}}\sqrt{\widetilde{n}_{t_\iota}}+\sigma_{t_\iota}\eta_{\text{inf}}\alpha_{t_\iota}\big)}\right\},
	\end{align}
	where $\sigma_{t_\iota}=\Vert I-\phi_{t_\iota}\Vert_2$, $\phi_{t_\iota}=\frac{\mathbf{1}_{\widetilde{n}_{t_\iota}}\mathbf{1}_{\widetilde{n}_{t_\iota}}^\top}{\widetilde{n}_{t_\iota}}$, $\beta_{t_\iota}=\max_{i\in{\mathcal{C}_{t_\iota}}}\{\beta_{i,t_\iota}\}$ and $\beta_{i,t_\iota}=\max_{v_{i,j}\in{\mathcal{V}_i^{t_\iota}}}\{\beta_{i,j,t_\iota}\}.$
\end{theorem}
\begin{proof}
	\indent The proof is presented in Appendix.
\end{proof}
The convergence upper bound
\[
\left\|\hat{\Gamma}_{t_{\iota+1}}^{\kappa_{t_\iota}-1}\right\|_F
\left(
\|e(t_\iota)\|_F
+
4\beta_{t_\iota}\widetilde{n}_{t_\iota}B_x
\right),
\]
where $\hat{\Gamma}_{t_{\iota+1}}$ is independent of
$\widetilde{n}_{t_\iota}$, presents the effect of the current population
size on the convergence error.
Furthermore, Theorem 2 shows that, under a
suitably selected fixed stepsize, the proposed algorithm achieves exponential
convergence. 
In particular, as the length of the interval between two
consecutive structural events increases,
\[
\lim_{\kappa_{t_\iota}\rightarrow\infty}
\|e(t_{\iota+1})\|_F=0.
\]
However, in a distributed network, it is generally difficult for agents to
locally determine a suitable range of fixed stepsizes through neighbor-to-
neighbor interactions. 
This motivates the use of diminishing stepsizes as an
alternative when a suitable fixed stepsize cannot be selected.

Accordingly, for each $v_{i,j}\in\mathcal{V}^{t_\iota}$, the stepsize
$\eta_{i,j}(t)$ is assumed to satisfy the following condition.
\begin{assumption}
	$\{\eta_{i,j}(t)\}$ is non-increasing, $\sum_{t=0}^{\infty}\eta_{i,j}(t)=\infty$ and $\sum_{t=0}^{\infty}\eta_{i,j}^2(t)<\infty$ for all $v_{i,j}$.
\end{assumption}
Let $E(t)\coloneqq \Vert Y(t)-\mathbf{1}_{\widetilde{n}_{t_\iota}}x_{t_\iota}^* \Vert_F$ denote the state error with respect to the Nash equilibrium $x_{t_\iota}^*$.
Then, the following corollary follows directly from Theorem 2.
\begin{corollary}
	Suppose Assumption 1 holds.
	If $\kappa_{t_\iota}>1$ and $\varepsilon<\frac{1}{n_{\text{tal}}}$, then, 
			\begin{align*}
		\Vert E(t_{\iota+1})\Vert_F\le  \Pi_{t=t_\iota+1}^{t_{\iota+1}-1} \Vert \hat{\Gamma}_{t}\Vert_F(\Vert e(t_\iota)\Vert_F+4\beta_{t_\iota} \widetilde{n}_{t_\iota}B_x).
	\end{align*}
\end{corollary}
\begin{proof}
	Then, it follows directly from Theorem 2 that Corollary 1 holds.
\end{proof}
Corollary 1 provides an upper bound on the convergence error of the proposed
algorithm over each inter-event interval under diminishing stepsizes.
If Assumptions 2, 3 hold, $\varepsilon<\frac{1}{n_{\text{tal}}}$, and $\kappa_{t_\iota}$ is sufficiently large, then according to Theorem 2, this upper bound will continuously shrink.
However, since
$
\lim_{t_{\iota+1}\rightarrow\infty}
\rho(\hat{\Gamma}_{t_{\iota+1}})=1,
$ the upper bound is not necessarily guaranteed to converge to zero.
Therefore, Corollary 1 only provides a relatively loose convergence bound and
does not by itself establish the asymptotic convergence of the proposed
algorithm.

To further characterize the asymptotic convergence of the proposed algorithm
under diminishing stepsizes, we establish the following theorem.
\begin{theorem}
	Suppose Assumptions 1-4 hold.
	If $\varepsilon<\frac{1}{n_{\text{tal}}}$, then 
	$$\lim\limits_{\kappa_{t_\iota}\rightarrow\infty}E(t_{\iota+1})=0.$$
\end{theorem}
\begin{proof}
	\indent The proof is presented in Appendix.
\end{proof}
According to Theorem 3, the algorithm converges as long as the stepsize sequence satisfies Assumption 4.
This also implies that, when a suitable fixed stepsize cannot be selected, a diminishing stepsize can be employed to achieve convergence of the proposed algorithm.

As we all know, there exist various stepsize sequences satisfying Assumption 4. 
Among them, the sequence $\{\eta_{i,j}(T_{i,j})/(t-T_{i,j}+1)\}$ is one of the most commonly considered choices in distributed optimization and NE seeking, where $a_0$ is a positive constant.
It is evident that, as time progresses, this kind stepsize gradually becomes smaller, which in turn slows down the state updates. 
This is particularly undesirable for agents that remain in the game throughout the entire process, as their stepsizes continue to decrease over time, resulting in increasingly slow state evolution.

Therefore, it is desirable to reduce the frequency of stepsize updates while preserving the convergence properties of the algorithm. To this end, we consider the following stepsize scheme.

First, for each agent, its time axis is partitioned into a countable sequence of consecutive intervals. 
Specifically, the $k$-th interval associated with agent $v_{i,j}$ is defined as
$
\left[\mathcal{T}_{i,j}(k),\mathcal{T}_{i,j}(k+1)\right),
$
where 
$
\mathcal{T}_{i,j}(k)=T_{i,j}+a^{k-1}-1,
$
and $a>1$ is integer.
Accordingly, the length of the $k$-th interval is given by
$
\mathcal{T}_{i,j}(k+1)-\mathcal{T}_{i,j}(k)=a^{k-1},
$
which increases exponentially with $k$.
We refer to $\left[\mathcal{T}{i,j}(k),\mathcal{T}{i,j}(k+1)\right)$ as the $k$-th exploration interval of agent $v_{i,j}$. During the $k$-th exploration interval, agent $v_{i,j}$ employs the following stepsize:
\begin{flalign}\label{7-1}
	\eta_{i,j}(t)=\frac{\eta_{i,j}(T_{i,j})}{a^{k-1}}, \quad t\in\left[\mathcal{T}_{i,j}(k),\mathcal{T}_{i,j}(k+1)\right),
\end{flalign}
where $\eta_{i,j}(T_{i,j})>0$ is a given constant. 
Accordingly, the stepsize is kept constant throughout each exploration interval and is reduced by half only when transitioning to the next interval. 
In this way, the frequency of stepsize updates decreases over time, while the stepsize gradually decreases as the exploration intervals become longer (see Table I).
%
\begin{table}[htbp] \centering 
	\caption{Comparison of the conventional and proposed stepsize schemes at different time instants ($a=2$)} \label{tab:stepsize_scheme} 
\begin{threeparttable}
	\begin{tabular}{c c c}
		\toprule
		Time instant $t$ & Stepsize $\frac{\eta_0}{t-T_{i,j}+1}$ & Proposed stepsize \eqref{7-1}\\
		\midrule
		$T_{i,j}$ & $\eta_0$ & $\eta_0$ \\
		$T_{i,j}+1$ & $\eta_0/2$ & $\eta_0/2$ \\
		$T_{i,j}+2$ & $\eta_0/3$ & $\eta_0/2$ \\
		$T_{i,j}+3$ & $\eta_0/4$ & $\eta_0/4$ \\
		$T_{i,j}+4$ & $\eta_0/5$ & $\eta_0/4$ \\
		$T_{i,j}+5$ & $\eta_0/6$ & $\eta_0/4$ \\
		$T_{i,j}+6$ & $\eta_0/7$ & $\eta_0/4$ \\
		$\vdots$ & $\vdots$ & $\vdots$ \\
		\bottomrule
	\end{tabular} 
		
		\begin{tablenotes}
			\footnotesize
			\item Note: $\eta_0$ denotes the initial stepsize, i.e., $\eta_0=\eta_{i,j}(T_{i,j})$.
		\end{tablenotes}
		\end{threeparttable}
\end{table}
%
Next, the following result holds easily.
\begin{corollary}
	Suppose Assumptions 1-3 hold.
	If $\varepsilon<\frac{1}{n_{\text{tal}}}$ and the stepsize satisfies the \eqref{7-1}, then
	$$\lim_{\kappa_{\iota}\rightarrow\infty}E(t_{\iota+1})=0.$$
\end{corollary}
\begin{proof}
	\indent First, it can be derive that 
	\begin{align*}
		\sum\limits_{t}^{\infty}\eta_{i,j}(t)=\sum\limits_{k}^{\infty}2^{k-1}\frac{\eta_{i,j}(T_{i,j})}{2^{k-1}}=\infty,
	\end{align*}
	and 
	$$
	\sum\limits_{t}^{\infty}\eta_{i,j}^2(t)=\sum\limits_{k}^{\infty}2^{k-1}\frac{\eta^2_{i,j}}{2^{2(k-1)}}<\infty.
	$$
	Next, it follows directly from Theorem 3 that Corollary 2 holds.
\end{proof}
Corollary 2 ensures the convergence of the proposed algorithm when the stepsize is selected according to \eqref{7-1}.

\section{Simulation}
To verify the effectiveness of the proposed distributed Nash equilibrium
seeking algorithm, an open connecting control game is considered as a
numerical example for illustration, which is the dynamic version of the connecting control game in \cite{chen2025non}.
In the game, the participation set, the communication
topologies, and the physical interaction relationships are allowed to vary
over time. 

\begin{table}[htbp]
	\centering
	\caption{Simulation parameters.}
	\label{tab:simulation_parameters}
	
	\renewcommand{\arraystretch}{1.3}  
	
	\setlength{\heavyrulewidth}{1pt}   
	\setlength{\lightrulewidth}{0.5pt} 
	
	\begin{tabular}{lc}
		\toprule
		Parameter & Value \\
		\midrule
		Initial number of coalitions $|\mathcal{C}_0|$ & $6$ \\
		Initial number of agents per coalition $n_{i,0}$ & $5$ \\
		Coalition entry probability $p_c$ & $0.15$ \\
		New coalition size & $1$--$3$ \\
		Agent entry probability $p_a^{\mathrm{in}}$ & $0.60$ \\
		Agent departure probability $p_a^{\mathrm{out}}$ & $0.75$ \\
		Feasible set $\Omega_{i,j}$ & $[-5,5]$ \\
		Mixing parameter $\varepsilon$ & $0.01$ \\
		\bottomrule
	\end{tabular}
\end{table}

Initially, the game consists of $
|\mathcal{C}_0|=6, n_{i,0}=5, |\mathcal{V}^{0}|=30 .
$
The maximum number of agents that the platform can accommodate is set as
$n_{\mathrm{tal}}=99$. 
The changes in the participating population are
triggered by random structural events. Specifically, the interval
$\kappa_{t_\iota}$ between two consecutive events follows a geometric
distribution:
\begin{equation*}
	\Pr(\kappa_{t_\iota}=k)
	=(1-p)^{k-1}p,\quad k=1,2,\ldots,
\end{equation*}
where $p=1/1000$, resulting in
$\mathbb{E}[\kappa_{t_\iota}]=1000$.
At each structural event, a new coalition is generated with probability
$p_c=0.15$, and its initial size is randomly selected from
$\{1,2,3\}$.
An individual agent joins an existing coalition with probability
$p_a^{\mathrm{in}}=0.60$, where the destination coalition is uniformly
selected from the current participating coalitions. 
Moreover, an agent
departure occurs with probability $p_a^{\mathrm{out}}=0.75$, where the leaving
agent is randomly selected from the participating agents. 
If the departing
agent is not the last member of its coalition, at least one intra-coalition
neighbor is required to remain.
 Otherwise, the departure is regarded as
infeasible. 
When the last agent leaves a coalition, the coalition naturally
disappears from the game. 
The evolution of the numbers of participating
agents and coalitions is illustrated in Figs. 2 and 3.

To characterize the dynamic communication and interaction relationships,
three types of graphs are reconstructed after each structural event, including the intra-coalition
communication topology, the inter-coalition communication topology, and the
physical interaction topology. 
The intra-coalition communication topology is
generated as a ring graph.
The inter-coalition communication topology is generated based on the
Erdős--Rényi random graph model, and only connected realizations are retained.
The physical interaction topology is independently generated to describe the
coupling relationships among agents.

\begin{figure}
	\centering
	\includegraphics[height=0.55\linewidth,width=0.73\linewidth]{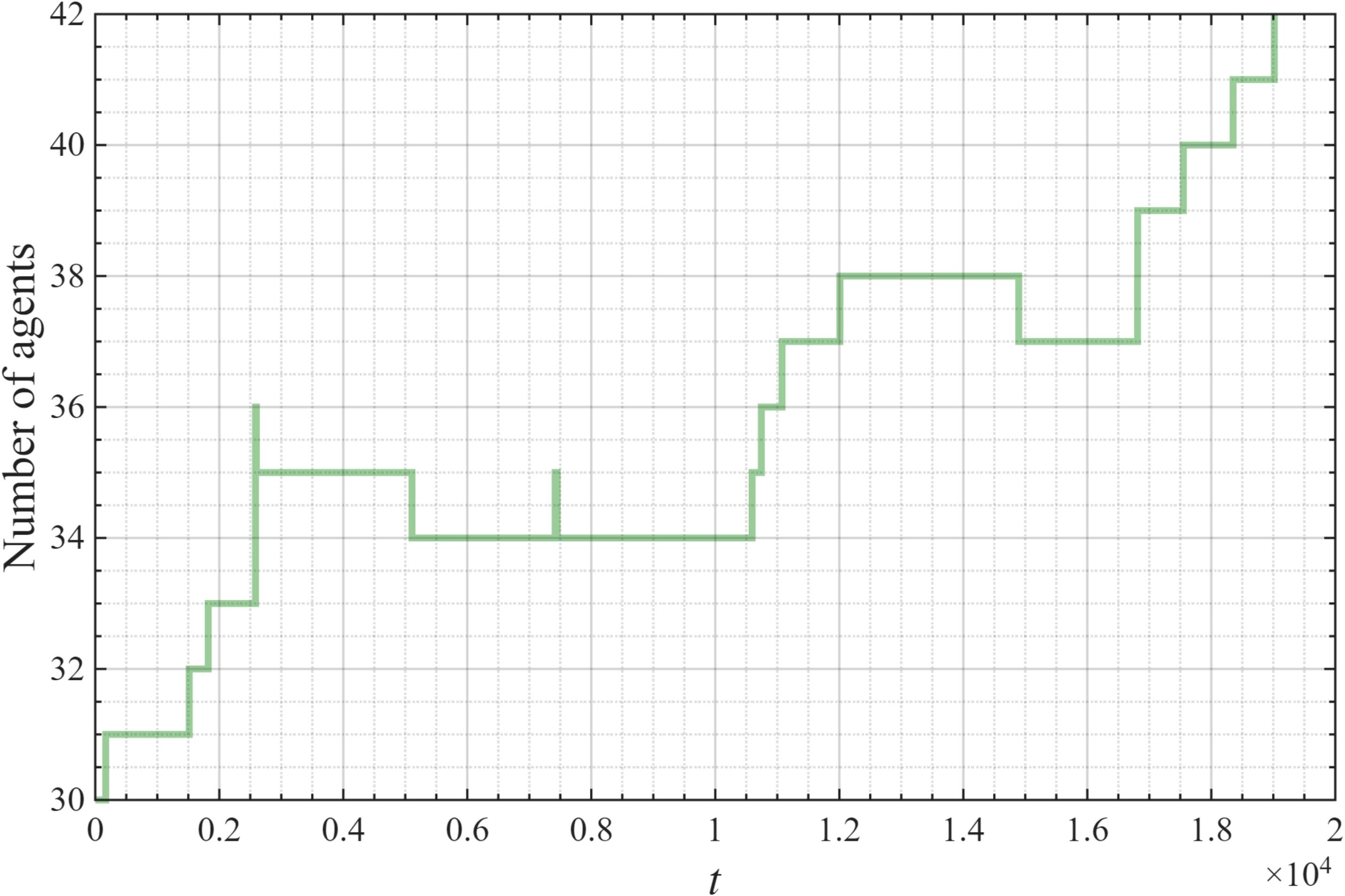}
	\captionsetup{font=footnotesize}
	\caption{The evolution of the number of participating agents over time.}
\end{figure}
\begin{figure}
	\centering
	\includegraphics[height=0.55\linewidth,width=0.73\linewidth]{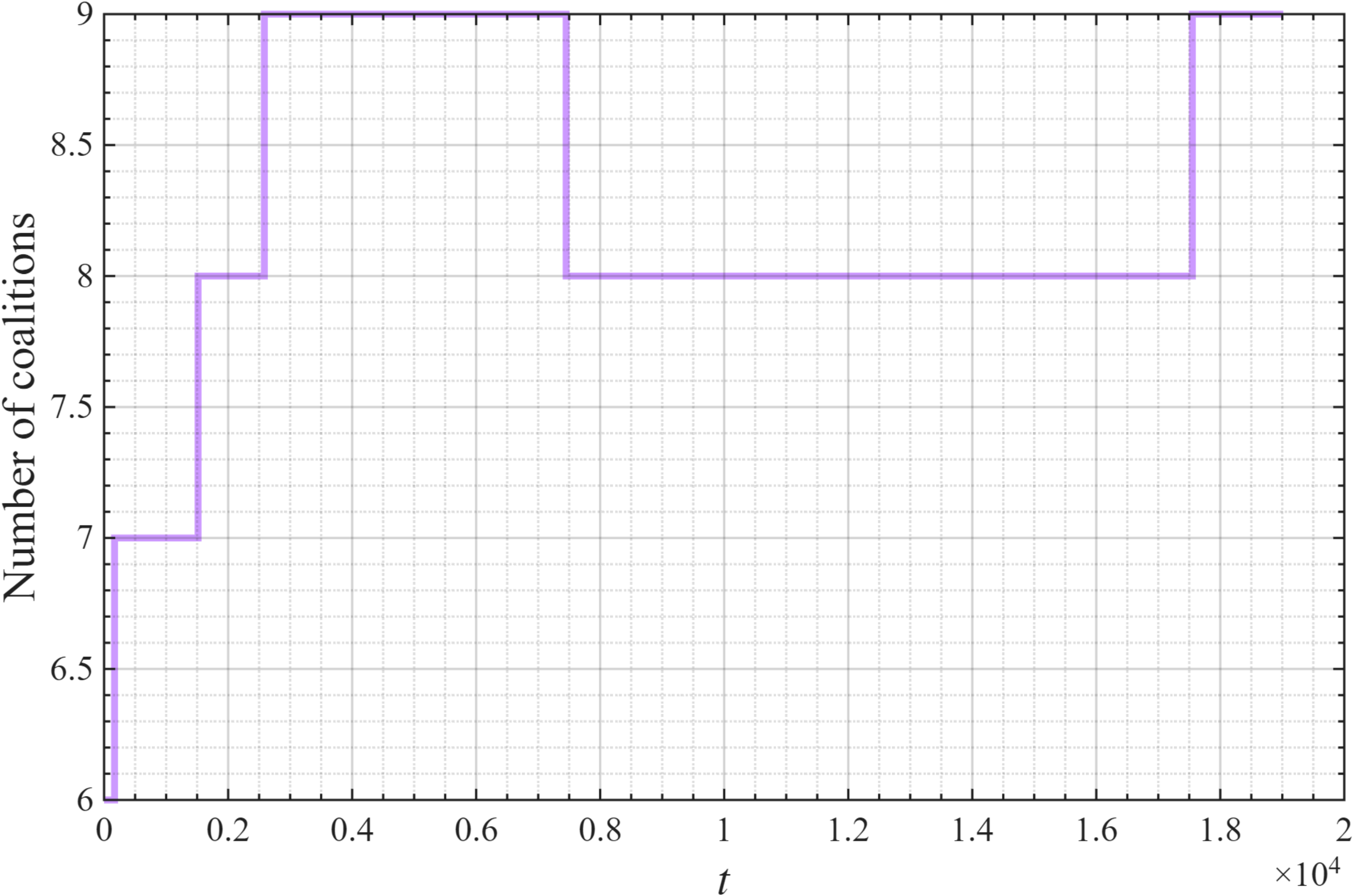}
	\captionsetup{font=footnotesize}
	\caption{The evolution of the number of participating coalitions over time.}
\end{figure}

In the game, for each agent $v_{i,j}$, the decision variable $x_{i,j}$ is constrained in
the compact convex set
$\Omega_{i,j}=[-5,5]$.
The local cost function of $v_{i,j}$ is defined
as $f_{i,j}^{t_\iota}(x)
	=
	r_i x_{i,j}^{2}
	+
	d_i x_{i,j}
	+
	b_i
	+
	\sum_{v_{h,l}\in\mathcal{S}_{i,j}^{t_\iota}}
	\left(
	x_{i,j}-x_{h,l}
	\right)^2,
$
where  $r_i$, $d_i$, and $b_i$ are constants randomly selected according to the uniform distribution, and $\mathcal{S}_{i,j}^{t_\iota}$ denotes the set of
physical interaction neighbors of agent $v_{i,j}$.
The first three terms
represent the individual cost, while the last term describes the coupling
effect induced by physical interactions.
The cost function of coalition $i$
is given by
$	f_i(x_i,x_{-i})
	=
	\sum_{v_{i,j}\in\mathcal{V}_{i}^{t}}
	f_{i,j}^{t}(x).
$
Since the Nash equilibrium varies with the participating population and
interaction relationships, it is recalculated at the initial time and after
each structural event.
In the simulation, numerical methods are employed to
obtain the Nash equilibrium $\hat{x}_{t_\iota}^{*}$, which is used as the
reference for evaluating the convergence performance. 
Then, the convergence
error is defined as
$$
\|e(t)\|_F
=
\left\|
\begin{bmatrix}
	Y(t)-\mathbf{1}_{\widetilde n_{t_\iota}}
	\hat{x}_{t_\iota}^{*\top}\\
	S(t)-\bar S(t)
\end{bmatrix}
\right\|_F .
$$
The algorithm parameter is set as $\varepsilon=0.01$.

To verify the effectiveness of the proposed mass-preserving gradient-tracking
technique, Fig. 4 presents the evolution of
$\|\bar S(t)-\bar Q_{t_\iota}(t)\|_F$ under the fixed stepsize
$\eta_{i,j}=0.01$. 
It can be observed that the tracking error increases when
structural events occur due to the changes in the participating population
and gradient information. 
After the events, the tracking error is immediately restored
to zero.
This result verifies the effectiveness of the proposed mass-preserving gradient-tracking technique for open dynamic systems and is consistent with the conclusion of Theorem 1.

\begin{figure}
	\centering
	\includegraphics[height=0.55\linewidth,width=0.73\linewidth]{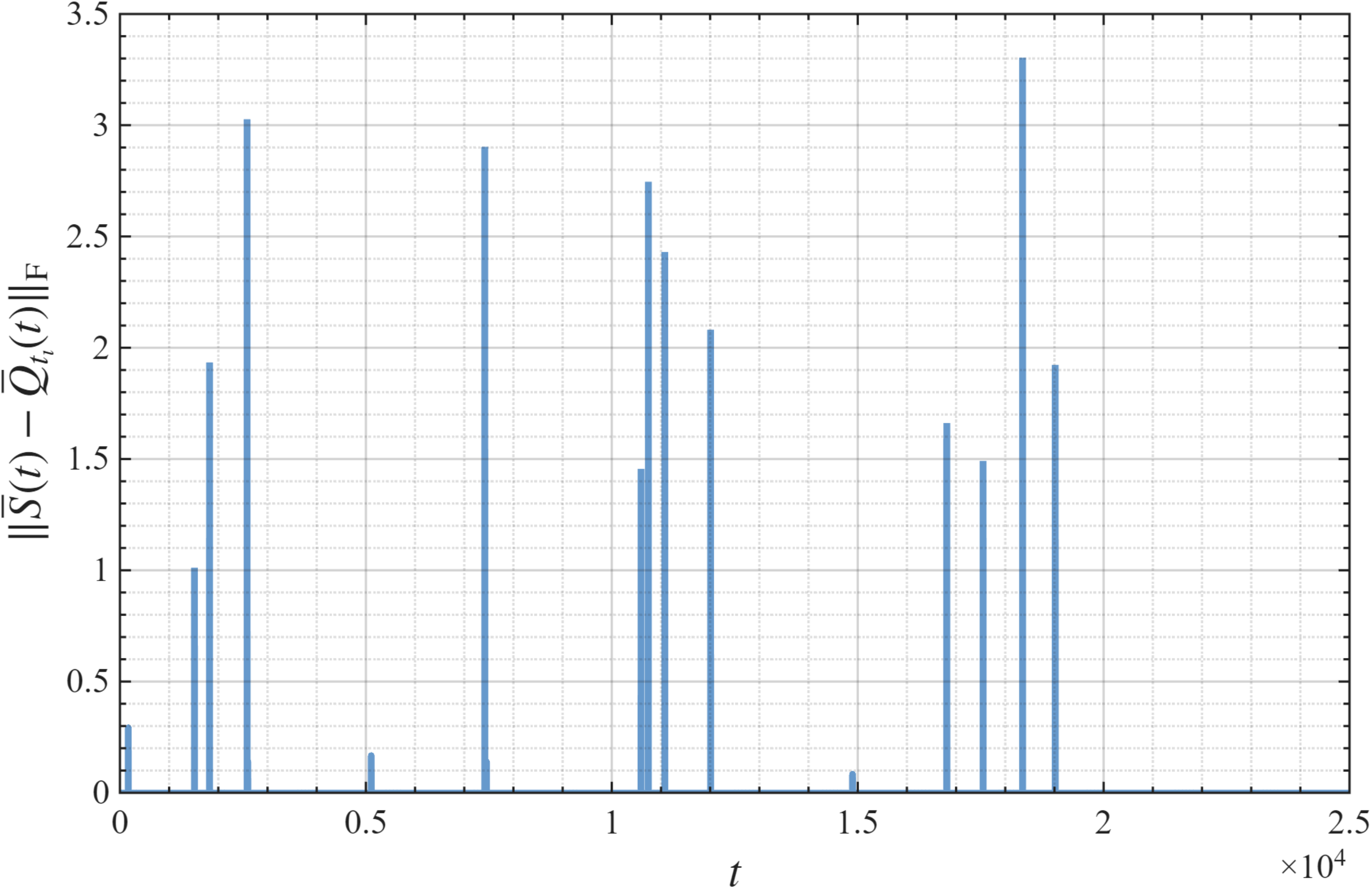}
	\captionsetup{font=footnotesize}
	\caption{The trajectory of $\Vert \bar{S}(t)-\bar{Q}_{t_\iota}(t)\Vert_F$ under the fixed stepsize $\eta_{i,j}=0.01$.}
\end{figure}

Fig. 5 compares the convergence performance of the proposed algorithm under
three different stepsize schemes: the fixed stepsize
$\eta_{i,j}=0.01$, the classical diminishing stepsize
$\frac{5}{t-T_{i,j}+1}$, and the exponential-interval diminishing stepsize
$$
\eta_{i,j}(t)=\frac{5}{2^{k-1}},
\quad
t\in[\mathcal{T}_{i,j}(k),\mathcal{T}_{i,j}(k+1)),
$$
where
$\mathcal{T}_{i,j}(k)=T_{i,j}+2^{k-1}-1 .
$
It can be seen that all three schemes drive the algorithm toward the Nash
equilibrium. 
In the early stage, the classical diminishing stepsize achieves
better convergence performance than the fixed stepsize scheme because of its
larger initial stepsize, while the exponential-interval diminishing stepsize
provides an intermediate performance. 
As time increases, the convergence
speed of the fixed stepsize scheme gradually becomes faster than those of the
two diminishing stepsize schemes. 
Moreover, the exponential-interval
diminishing stepsize converges faster than the classical diminishing
stepsize, since its stepsize remains unchanged within each exploration
interval and decreases only when the interval is enlarged. 
Therefore, it
maintains a larger effective stepsize for a longer period. 
Finally, the error
curve under the fixed stepsize scheme exhibits an almost constant decay rate,
whereas the decay rates of both diminishing stepsize schemes gradually
decrease over time. 
This observation agrees with the stepsize analysis
presented previously.

\begin{figure}
	\centering
	\includegraphics[height=0.55\linewidth,width=0.73\linewidth]{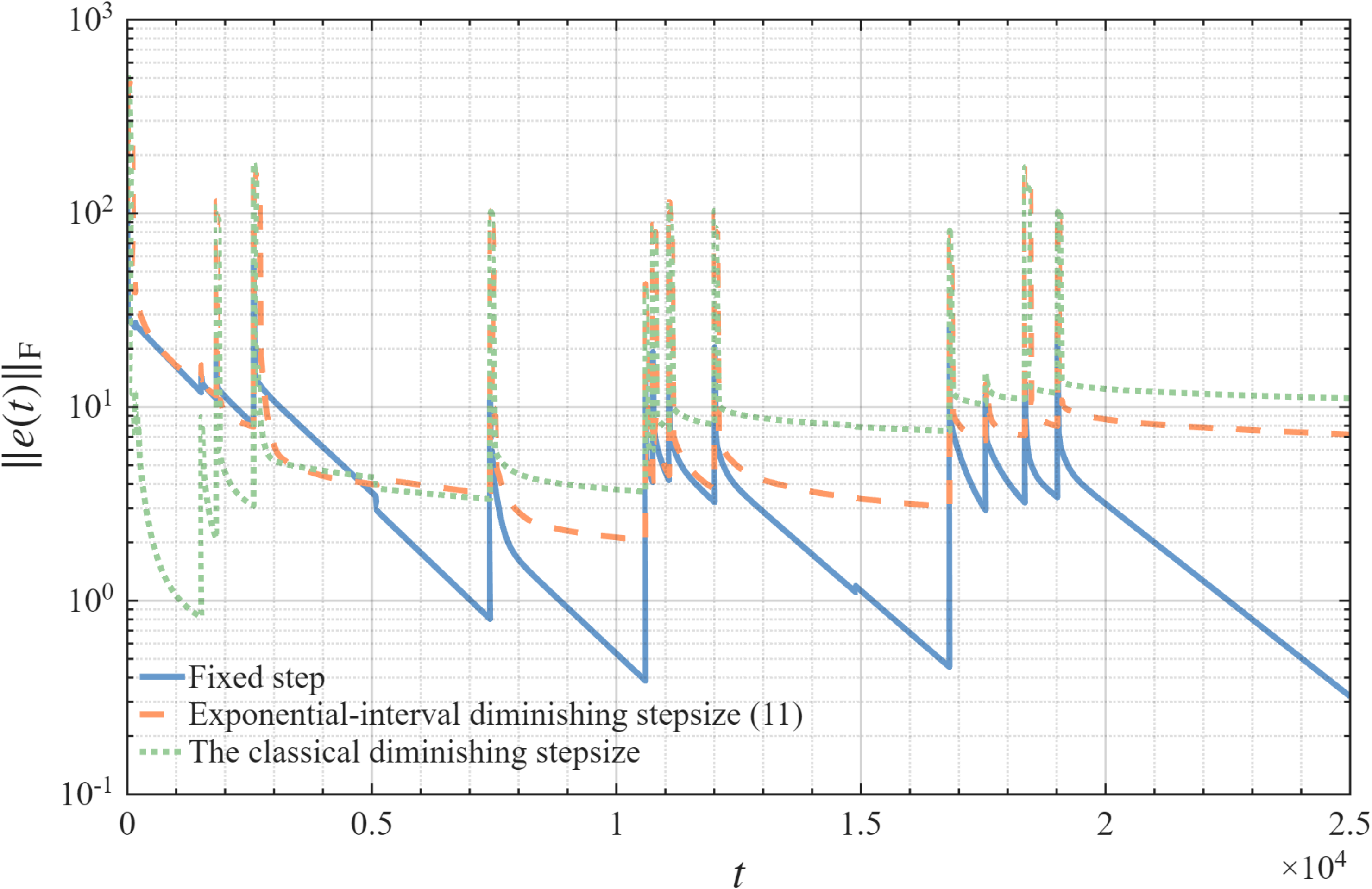}
	\captionsetup{font=footnotesize}
	\caption{The trajectories of $\Vert e(t)\Vert_F$ under different stepsizes.}
\end{figure}



\section{Conclusion}
To address the loss of gradient information induced by agent departures,
a mass-preserving gradient-tracking mechanism is developed for open
multi-coalition games. 
By redistributing the gradient information
associated with departing agents to their remaining neighbors within the same
coalition, the proposed mechanism preserves the aggregate gradient information
of the coalition objective and maintains gradient tracking under dynamic
population changes without imposing additional initialization conditions beyond
those required by conventional gradient-tracking methods. 
Since the
multi-coalition game reduces to distributed optimization when the number of
coalitions is one, the proposed mechanism also provides a gradient-tracking
scheme for open distributed optimization.

Building on the proposed mechanism, a distributed Nash equilibrium seeking
framework is developed for open multi-coalition games.
Convergence conditions
are established for both agent-dependent fixed stepsizes and diminishing
stepsizes, and a convergence error bound is derived for each interval between
two consecutive structural events.
To further mitigate the convergence
slowdown caused by the continuously decreasing stepsize in the classical
scheme $\eta_0/t$, an exponential-interval diminishing stepsize strategy is
introduced. 
Specifically, each agent's time axis is partitioned into a
sequence of exploration intervals with increasing lengths, over which the
stepsize remains unchanged and is reduced only at interval boundaries.
 This
design reduces the frequency of stepsize changes and alleviates the
later-stage convergence degradation associated with classical diminishing
stepsizes.

\appendix
\renewcommand{\proofname}{\bf Proof of Theorem 1}
\begin{proof}
	If $v_{k,m}\notin{(\mathcal{N}_{i,j}^{i,t_{\iota}}\cup\mathcal{V}_i^{t_{\iota}}\cap\mathcal{U}_{i,j}^{t_{\iota}})}$, and $v_{i,j}$ does not maintain $s_{i,j}^{m}(t)$, then it is true that $q_{i,j}^{m,t_{\iota}}(t)=q_{i,j}^{m,t_{\iota}}(t)=0$.
    When $s_{i,j}^{m}(t_\iota)$ does not exist, let's set $s_{i,j}^{m}(t_\iota)=0$.
	Then, \eqref{1-3} can be equivalently transformed into the following form:
	\begin{flalign}\label{5-2}
		&s_{i,j}^m(t+1)=s_{i,j}^m(t)+\varepsilon\sum\limits_{v_{i,l}\in{\mathcal{N}_{i,j}^{i,t_{\iota}}}}[\mathcal{A}^i]_j^l(s_{i,l}^m(t)-s_{i,j}^m(t))&\notag\\
		&+q_{i,j}^{m,t_{\iota}}(t+1)-q_{i,j}^{m,t_{\iota}}(t).
	\end{flalign}
Let $\hat{s}_{i,j}(t_\iota)=s_{i,j}(t_\iota)+\sum_{v_{i,l}\in{\mathcal{D}_{t_\iota}\cap\mathcal{N}_{i,j}^{i,t_{\iota-1}}}}Mass_{i,l}.$
Then, we have
	\begin{flalign*}
		&\sum\limits_{v_{i,j}\in{\mathcal{R}_i^{t_{\iota}}}}\hat{s}_{i,j}(t_\iota)&\notag\\
		&=\sum\limits_{v_{i,j}\in{\mathcal{R}_i^{t_{\iota}}}}s_{i,j}(t_\iota)+\sum\limits_{v_{i,l}\in{\mathcal{V}_i^{t_{\iota-1}}}\cap\mathcal{D}_{t_\iota}}(s_{i,l}(t_\iota)-q^{t_{\iota-1}}_{i,l}(t_\iota))\notag\\
		&=\sum\limits_{v_{i,j}\in{\mathcal{V}_i^{t_{\iota-1}}}}s_{i,j}(t_\iota)-\sum\limits_{v_{i,l}\in{\mathcal{V}_i^{t_{\iota-1}}}\cap\mathcal{D}_{t_\iota}}q^{t_{\iota-1}}_{i,l}(t_\iota)\notag\\
		&=\sum\limits_{v_{i,j}\in{\mathcal{R}_i^{t_{\iota}}}}q_{i,j}^{t_{\iota-1}}(t_\iota).
\end{flalign*}

After the assignment operation \eqref{2-1}, it can be obtained that
\begin{flalign*}
&\sum\limits_{v_{i,j}\in{\mathcal{R}_i^{t_{\iota}}}}s_{i,j}(t_\iota)=\sum\limits_{v_{i,j}\in{\mathcal{R}_i^{t_{\iota}}}}\hat{s}_{i,j}(t_\iota)=\sum\limits_{v_{i,j}\in{\mathcal{R}_i^{t_{\iota}}}}q_{i,j}^{t_{\iota-1}}(t_\iota).
\end{flalign*}
According to $q_{i,j}^{t_{\iota-1}}(t_{\iota})=q_{i,j}^{t{\iota}}(t_\iota)$, $\forall v_{i,j}\in{\mathcal{V}_i^{t_\iota}-\mathcal{V}_i^{t_{\iota-1}}}$, it follows that
\begin{flalign*}
&\sum\limits_{v_{i,j}\in{\mathcal{V}_i^{t_\iota}}}s_{i,j}(t_\iota)=\sum\limits_{v_{i,j}\in{\mathcal{V}_i^{t_\iota}}}q_{i,j}^{t_{\iota-1}}(t_\iota).
\end{flalign*}
By \eqref{5-1} and \eqref{5-2}, it can be obtained that
\begin{flalign*}
	&\sum\limits_{v_{i,j}\in{\mathcal{V}_i^{t_{\iota}}}}s_{i,j}(t_\iota+1)&\\
	&=\sum\limits_{v_{i,j}\in{\mathcal{V}_i^{t_{\iota}}}}s_{i,j}(t_\iota)+\sum\limits_{v_{i,j}\in{\mathcal{V}_i^{t_{\iota}}}}q_{i,j}^{t_\iota}(t_\iota+1)-\sum\limits_{v_{i,j}\in{\mathcal{V}_i^{t_{\iota}}}}q^{t_{\iota-1}}_{i,j}(t_\iota)\\
	&=\sum\limits_{v_{i,j}\in{\mathcal{V}_i^{t_{\iota}}}}q_{i,j}^{t_\iota}(t_\iota+1).
\end{flalign*}
Then, it can be easily derived that
\begin{flalign*}
	&\sum\limits_{v_{i,j}\in{\mathcal{V}_i^{t_{\iota}}}}s_{i,j}(t)=\sum\limits_{v_{i,j}\in{\mathcal{V}_i^{t_{\iota}}}}q_{i,j}^{t_\iota}(t),
\end{flalign*}
for $t_\iota< t\le t_\iota-1$.
Hence, the required property is obtained.
\end{proof}


\renewcommand{\proofname}{\bf Proof of Theorem 2}
\begin{proof}
	First, let's consider the case $t\in{(t_\iota,t_{\iota+1})}$.
	
	For $v_{i,j},v_{k,m}\in{\mathcal{V}^{t_\iota}}$, if $y_{i,j}^{k,m}(t_\iota)$ is unavailable, let's set $y_{i,j}^{k,m}(t_\iota)=0$.
	Then, \eqref{105-1} and \eqref{1-2} can be equivalently transformed into the following forms respectively:
	\begin{flalign*}
		&y_{i,j}^{k,m}(t+1)= y_{i,j}^{k,m}(t)+\varepsilon\!\!\!\!\sum\limits_{v_{h,l}\in{\mathcal{N}_{i,j}^{t_{\iota}}}} [\mathcal{A}_{t_\iota}]_{\tau_{i,j}}^{\tau_{h,l}}(y_{h,l}^{k,m}(t)-y_{i,j}^{k,m}(t))&\notag\\
		&-\eta_{i,j}s^m_{i,j}(t),
	\end{flalign*}
	and 
	\begin{flalign*}
		&y_{i,j}^{k,m}(t+1)= y_{i,j}^{k,m}(t)+\varepsilon\!\!\!\!\sum\limits_{v_{h,l}\in{\mathcal{N}_{i,j}^{t_{\iota}}}} [\mathcal{A}_{t_\iota}]_{\tau_{i,j}}^{\tau_{h,l}}(y_{h,l}^{k,m}(t)-y_{i,j}^{k,m}(t)).
	\end{flalign*}

	Let $P_{t_\iota}=I_{\widetilde{n}_{t_\iota}}-\varepsilon L_{t_\iota}$, where $L_{t_\iota}$ is the Laplacian matrix associated with $\mathcal{A}^{t_\iota}$.
	Let $\hat{P}_{t_\iota}=\mathrm{blk}\{P_{i,t_\iota}\}_{i\in{\mathcal{C}_{t_\iota}}}$, where $P_{i,t_\iota}=I_{n_{i,t_\iota}}-\varepsilon L_{i,t_\iota}$ and $L_{i,t_\iota}$ is the Laplacian matrix associated with $\mathcal{A}_i^{t_\iota}$.
	Let $\hat{\Omega}_{i,t\iota}=\{z\in{\mathbb{R}^{\widetilde{n}_t}|e_{t_\iota}z\in{\Omega_i}}\}$, where $e_{t_\iota}=\mathrm{blk}(\mathcal{O}_{n_{1,t_\iota}},\dots,\mathcal{O}_{n_{i-1,t_\iota}},I_{n_{i,t_\iota}}\mathcal{O}_{n_{i+1,t_\iota}},\dots\mathcal{O}_{n_{N_{t_\iota},t_\iota}})$.
	For $t_\iota$, define a map $\boldsymbol{\mathrm{Proj}}_{\Omega_{t_\iota}}:\mathbb{R}^{\widetilde{n}_t\times \widetilde{n}_t}\rightarrow \mathbb{R}^{\widetilde{n}_t\times \widetilde{n}_t} $ $\boldsymbol{\mathrm{Proj}}_{\Omega_{t_\iota}}(A)$, where $[\boldsymbol{\mathrm{Proj}}_{\Omega_{t_\iota}}(A)]_{\tau_{i,j}}=\mathrm{Proj}_{\hat{\Omega}_{i,t_\iota}}([A]_{\tau_{i,j}})$, $\forall A\in{\mathbb{R}^{\widetilde{n}_t\times \widetilde{n}_t}}$.
	Then, the update process of the algorithm can be characterized by the following iterative equation:
	\begin{equation}\label{6-7}
		\begin{cases}
			Y(t+1)=\boldsymbol{\mathrm{Proj}}_{\Omega_{t_\iota}}\left(P_{t_\iota}Y(t)-\boldsymbol{\eta}S(t)\right),\\
			S(t+1)=\hat{P}_{t_\iota}S(t)+Q_{t_\iota}(t+1)-	Q_{t_\iota}(t),
		\end{cases}
	\end{equation}
	where $\boldsymbol{\eta}=\mathrm{blk}\{\eta_i\}_{i\in{\mathcal{C}_{t_\iota}}}$ and $\eta_i=\mathrm{col}\{\eta_{i.j}\}_{v_{i,j}\in{\mathcal{V}_i^{t_\iota}}}$.
	
	According to the non-expansiveness property of projection, we have
	\begin{flalign}\label{5-3}
		&\Vert Y(t+1)-\mathbf{1}_{\widetilde{n}_{t_\iota}}x_{t_\iota}^*\Vert_F&\notag\\
		&=\left\Vert \boldsymbol{\mathrm{Proj}}_{\Omega_{t_\iota}}\left(P_{t_\iota}Y(t)-\boldsymbol{\eta}S(t)\right)\right.\notag\\
		&\left.-\boldsymbol{\mathrm{Proj}}_{\Omega_{t_\iota}}\left(\mathbf{1}_{\widetilde{n}_{t_\iota}}x_{t_\iota}^*-\boldsymbol{\eta}G_{t_\iota}(x_{t_\iota}^*)\right)\right\Vert_F\notag\\
		&\le\left\Vert P_{t_\iota}Y(t)-\mathbf{1}_{\widetilde{n}_{t_\iota}}x_{t_\iota}^*-\mathbf{\eta}\big(\bar{S}(t)-G_{t_\iota}(x_{t_\iota}^*)\big)\right\Vert_F\notag\\
		&+\eta_{\text{max}}\Vert S(t)-\bar{S}(t)\Vert_F,
	\end{flalign}
	where $\bar{S}(t)=\mathrm{blk}\{\bar{S}_i(t)\}_{i\in{\mathcal{C}_{t_\iota}}}$, $\bar{S}(t)=\frac{\mathbf{1}_{n_{i,t_\iota}}\mathbf{1}_{n_{i,t_\iota}}^\top}{n_{i,t_\iota}}S_i(t)$,  $G_{t_\iota}(x_{t_\iota}^*)=\mathrm{blk}\{\frac{\mathbf{1}_{n_{i,t_\iota}}}{n_{i,t_\iota}}g_{i,t_\iota}(x_{t_\iota}^*)\}_{i\in{\mathcal{C}_{t_\iota}}}$, and $\eta_{\text{max}}=max_{v_{i,j}\in{\mathcal{V}^{t_\iota}}}\{\eta_{i,j}\}$.
	
	First, consider the first item on the right-hand side of \eqref{5-3}. 
	It can be obtained that 
	\begin{flalign}\label{6-1}
		&\left\Vert P_{t_\iota}Y(t)-\mathbf{1}_{\widetilde{n}_{t_\iota}}x_{t_\iota}^*-\boldsymbol{\eta}\big(\bar{S}(t)-G_{t_\iota}(x_{t_\iota}^*)\big)\right\Vert_F^2&\notag\\
		&\le \Vert \phi_{t_\iota}Y(t)-\mathbf{1}_{\widetilde{n}_{t_\iota}}x_{t_\iota}^*\Vert_F^2+\Vert P_{t_\iota}Y_\perp(t)\Vert_F^2\notag\\
		&+\eta_{\mathrm{max}}^2\Vert \bar{S}(t)-G_{t_\iota}(x_{t_\iota}^*)\Vert_F\notag\\
		&-2\left<\phi_{t_\iota}Y(t)-\mathbf{1}_{\widetilde{n}_{t_\iota}}x_{t_\iota}^*,\boldsymbol{\eta}\big(\bar{S}(t)-G_{t_\iota}(x_{t_\iota}^*)\big)\right>\notag\\
		&-2\left<P_{t_\iota}Y_\perp(t),\boldsymbol{\eta}\big(\bar{S}(t)-G_{t_\iota}(x_{t_\iota}^*)\big)\right>,
	\end{flalign}
	where $\phi_{t_\iota}=\frac{\mathbf{1}_{\widetilde{n}_{t_\iota}}\mathbf{1}_{\widetilde{n}_{t_\iota}}^\top}{\widetilde{n}_{t_\iota}}$, and $Y_\perp(t)=(I_{\widetilde{n}_{t_\iota}}-\phi_{t_\iota})Y(t)$.
	By the definition of $Y_\perp(t)$, we have
	\begin{flalign}\label{6-2}
		&\Vert P_{t_\iota}Y_\perp(t)\Vert_F^2&\notag\\
		&=\Vert P_{t_\iota}(I_{\widetilde{n}_{t_\iota}}-\phi_{t_\iota})Y(t)\Vert_F\notag\\
		&=\Vert (P_{t_\iota}-\phi_{t_\iota})(Y(t)-\phi_{t_\iota}Y(t))\Vert_F\notag\\
		&\le \Vert P_{t_\iota}-\phi_{t_\iota}\Vert_2\Vert Y(t)-\phi_{t_\iota}Y(t) \Vert_F\notag\\
		&=\sigma_{t_\iota}\Vert Y_\perp(t)\Vert_F,
	\end{flalign}
	where $\sigma_{t_\iota}=\Vert P_{t_\iota}-\phi_{t_\iota}\Vert_2$.
	According to the property of Frobenius norm and Assumption 1, it can be readily derived that
	\begin{flalign}\label{6-3}
		&\Vert \bar{S}(t)-G_{t_\iota}(x_{t_\iota}^*)\Vert_F&\notag\\
		&=\begin{Vmatrix}
			\mathrm{col}\{\bar{S}_i(t)-\frac{\mathbf{1}_{n_{i,t_\iota}}}{n_{i,t_\iota}}g_i^{t_\iota}(x_{t_\iota}^*)\}_{i\in{\mathcal{C}_{t_\iota}}}
		\end{Vmatrix}_F\notag\\
		&=\left\Vert\mathrm{col}\left\{\frac{\mathbf{1}_{n_{i,t_\iota}}}{n_{i,t_\iota}}\sum\limits_{v_{i,j}\in{\mathcal{V}_i^{t_\iota}}}\bigg(\nabla_i f_{i,j}^{t_\iota}\big(y_{i,j}(t)\big)-f_{i,j}^{t_\iota}\big(x_{t_\iota}^*\big)\bigg)\right\}_{i\in{\mathcal{C}_{t_\iota}}}\right\Vert_F\notag\\
		&\le\Vert \mathrm{col}\{\beta_{i,t_\iota}\Vert Y_i(t)-\mathbf{1}_{n_{i,t_\iota}}x_{t_\iota}^*\Vert_F\}_{i\in{\mathcal{C}_{t_\iota}}}\Vert_F\notag\\
		&\le\beta_{t_\iota}\Vert Y(t)-\mathbf{1}_{\widetilde{n}_{t_\iota}}x_{t_\iota}^*\Vert_F,
	\end{flalign}
	where $\beta_{i,t_\iota}=\max_{v_{i,j}\in{\mathcal{V}_i^{t_\iota}}}\{\beta_{i,j,t_\iota}\}$ and $\beta_{t_\iota}=\max_{i\in{\mathcal{C}_{t_\iota}}}\{\beta_{i,t_\iota}\}.$
	For the fourth item of the \eqref{6-1}, we have
	\begin{flalign}\label{6-4}
		&-2\left<\phi_{t_\iota}Y(t)-\mathbf{1}_{\widetilde{n}_{t_\iota}}x_{t_\iota}^*,\boldsymbol{\eta}\big(\bar{S}(t)-G_{t_\iota}(x_{t_\iota}^*)\big)\right>&\notag\\
		&=-2\left<\phi_{t_\iota}Y(t)-\mathbf{1}_{\widetilde{n}_{t_\iota}}x_{t_\iota}^*,\boldsymbol{\eta}\big(\bar{S}(t)-G_{t_\iota}(\bar{y}(t))\big)\right>\notag\\
		&-2\left<\phi_{t_\iota}Y(t)-\mathbf{1}_{\widetilde{n}_{t_\iota}}x_{t_\iota}^*,\boldsymbol{\eta}\big(G_{t_\iota}(\bar{y}(t))-G_{t_\iota}(x_{t_\iota}^*)\big)\right>\notag\\
		&\le 2\eta_{\text{max}}\Vert \phi_{t_\iota}Y(t)-\mathbf{1}_{\widetilde{n}_{t_\iota}}x_{t_\iota}^*\Vert_F\Vert \bar{S}(t)-G_{t_\iota}\big(\bar{y}(t)\big)\Vert_F\notag\\
		&-2\frac{\alpha_{t_\iota}\eta_{\text{min}}}{\sqrt{\widetilde{n}_{t_\iota}}}\Vert \phi_{t_\iota}Y(t)-\mathbf{1}_{\widetilde{n}_{t_\iota}}x_{t_\iota}^* \Vert_F\notag\\
		&\le 2\beta_{t_\iota}\eta_{\text{max}}\Vert \phi_{t_\iota}Y(t)-\mathbf{1}_{\widetilde{n}_{t_\iota}}x_{t_\iota}^*\Vert_F\Vert Y(t)-\phi_{t_\iota}Y(t)\Vert_F\notag\\
		&-\frac{2\alpha_{t_\iota}\eta_{\text{min}}}{\sqrt{n_{tal}}}\Vert \phi_{t_\iota}Y(t)-\mathbf{1}_{\widetilde{n}_{t_\iota}}x_{t_\iota}^* \Vert_F,
	\end{flalign}
	where $\bar{y}(t)=\frac{\mathbf{1}^\top_{\widetilde{n}_{t_\iota}}Y(t)}{\widetilde{n}_{t_\iota}}$ and $\eta_{\text{min}}=\min_{v_{i,j}\in{\mathcal{V}_i^{t_\iota}}}\{\eta_{i,j}\}$.
	Similarly, we can obtain
	\begin{align}
		\label{6-5}
		&-2\left<P_{t_\iota}Y_\perp(t),\boldsymbol{\eta}\big(\bar{S}(t)-G_{t_\iota}(x_{t_\iota}^*)\big)\right>\notag&\\
		&\le 2\Vert P_{t_\iota}Y_\perp(t)\Vert_F\Vert \bar{S}(t)-G_{t_\iota}(x_{t_\iota}^*)\Vert_F\notag\\
		&\le 2\beta_{t_\iota}\Vert P_{t_\iota}Y_\perp(t)\Vert_F\Vert Y(t)-\mathbf{1}_{\widetilde{n}_{t_\iota}}x_{t_\iota}^*\Vert_F\notag\\
		&\le 2\beta_{t_\iota}\sigma_{t_\iota}\Vert Y_\perp(t)\Vert_F(\Vert \phi_{t_\iota}Y(t)-\mathbf{1}_{\widetilde{n}_{t_\iota}}x_{t_\iota}^*\Vert_F+\Vert Y_\perp(t)\Vert_F).
	\end{align}
   Then, substituting \eqref{6-2}-\eqref{6-5} into \eqref{6-1}, it follows that
   \begin{align}\label{6-6}
   	&\left\Vert P_{t_\iota}Y(t)-\mathbf{1}_{\widetilde{n}_{t_\iota}}x_{t_\iota}^*-\boldsymbol{\eta}\big(\bar{S}(t)-G_{t_\iota}(x_{t_\iota}^*)\big)\right\Vert_F^2&\notag\\
   	&\le \Vert \phi_{t_\iota}Y(t)-\mathbf{1}_{\widetilde{n}_{t_\iota}}x_{t_\iota}^*\Vert_F^2+\sigma_{t_\iota}^2\Vert Y_\perp(t)\Vert_F^2\notag\\
   	&+\eta_{\text{max}}^2\beta_{t_\iota}^2\Vert Y(t)-\mathbf{1}_{\widetilde{n}_{t_\iota}}x_{t_\iota}^*\Vert_F^2\notag\\
   	&+2\beta_{t_\iota}\eta_{\text{max}}\Vert \phi_{t_\iota}Y(t)-\mathbf{1}_{\widetilde{n}_{t_\iota}}x_{t_\iota}^*\Vert_F\Vert Y(t)-\phi_{t_\iota}Y(t)\Vert_F\notag\\
   	&-\frac{2\alpha_{t_\iota}\eta_{\text{min}}}{\sqrt{n_{tal}}}\Vert \phi_{t_\iota}Y(t)-\mathbf{1}_{\widetilde{n}_{t_\iota}}x_{t_\iota}^* \Vert_F^2\notag\\
   	&+2\beta_{t_\iota}\sigma_{t_\iota}\Vert Y_\perp(t)\Vert_F(\Vert \phi_{t_\iota}Y(t)-\mathbf{1}_{\widetilde{n}_{t_\iota}}x_{t_\iota}^*\Vert_F+\Vert Y_\perp(t)\Vert_F)\notag\\
   	&\le(\Vert\phi_{t_\iota}Y(t)-\mathbf{1}_{\widetilde{n}_{t_\iota}}x_{t_\iota}^* \Vert_F, \Vert Y_\perp(t)\Vert_F)
   	\Gamma_{t_\iota}
   	\begin{bmatrix}
   		\Vert\phi_{t_\iota}Y(t)-\mathbf{1}_{\widetilde{n}_{t_\iota}}x_{t_\iota}^* \Vert_F\\
   		\Vert Y_\perp(t)\Vert_F
   	\end{bmatrix}\notag\\
   	&\le \rho(\Gamma_{t_\iota})\Vert Y(t)-\mathbf{1}_{\widetilde{n}_{t_\iota}}x_{t_\iota}^*\Vert_F^2,
   \end{align}
   where 
   $$\Gamma_{t_\iota}=\begin{bmatrix}
   	1+\eta_{\text{max}}^2\beta_{t_\iota}^2-\frac{2\alpha_{t_\iota}\eta_{\text{min}}}{\sqrt{n_{tal}}}&\beta_{t_\iota}\eta_{\text{max}}(1+\sigma_{t_\iota})\\
   	\beta_{t_\iota}\eta_{\text{max}}(1+\sigma_{t_\iota})&\sigma_{t_\iota}^2+\eta_{\text{max}}^2\beta_{t_\iota}^2+2\beta_{t_\iota}\eta_{\text{max}}\sigma_{t_\iota}
   \end{bmatrix}$$ and 
   $\rho(\Gamma_{t_\iota})$ is the Spectral radius of $\Gamma_{t_\iota}$.
   Therefore, 
   \begin{align}
   	\label{6-8}
   	&\Vert Y(t+1)-\mathbf{1}_{\widetilde{n}_{t_\iota}}x_{t_\iota}^*\Vert_F&\notag\\
   	&\le \sqrt{\rho(\Gamma_{t_\iota})}\Vert Y(t)-\mathbf{1}_{\widetilde{n}_{t_\iota}}x_{t_\iota}^*\Vert_F+\eta_{\text{max}}\Vert S(t)-\bar{S}(t)\Vert_F.
   \end{align}
   
   Furthermore, it can be derived that
   \begin{align}\label{6-9}
   	&\Vert S(t+1)-\bar{S}(t+1)\Vert_F&\notag\\
   	&=\Vert S(t+1)-\hat{\phi}_{t_\iota}S(t+1)\Vert_F\notag\\
   	&=\Vert \hat{P}_{t_\iota}S(t)+Q_{t_\iota}(t+1)-Q_{t_\iota}(t)\notag\\
   	&-\hat{\phi}_{t_\iota}\big(S(t)+Q_{t_\iota}(t+1)-Q_{t_\iota}(t)\big)\Vert_F\notag\\
   	&\le\Vert \hat{P}_{t_\iota} -\hat{\phi}_{t_\iota}\Vert_2\Vert S(t)-\bar{S}(t) \Vert_F\notag\\
   	&+\beta_{t_\iota}\Vert I_{\widetilde{n}_{t_\iota}}-\hat{\phi}_{t_\iota}\Vert_2\Vert Y(t+1)-Y(t)\Vert_F\notag\\	
   	&\le \hat{\sigma}_{t_\iota}\Vert S(t)-\bar{S}(t)\Vert_F\notag\\
   	&+\beta_{t_\iota}\big(\Vert  Y(t+1)-\mathbf{1}_{\widetilde{n}_{t_\iota}}x_{t_\iota}^* \Vert_F
   	+\Vert  Y(t)-\mathbf{1}_{\widetilde{n}_{t_\iota}}x_{t_\iota}^* \Vert_F\big).
    \end{align} 
    where $\hat{\phi}_{t_\iota}=\mathrm{blk}\left\{\frac{\mathbf{1}_{n_{i,t_\iota}}\mathbf{1}^\top_{n_{i,t_\iota}}}{n_{i,t_\iota}}\right\}_{i\in{\mathcal{C}_{t_\iota}}}$ and $ \hat{\sigma}_{t_\iota}=\Vert \hat{P}_{t_\iota} -\hat{\phi}_{t_\iota}\Vert_2.$
    
    It follows from \eqref{6-8} and \eqref{6-9} that
    $e(t+1)\le \hat{\Gamma}_{t_\iota} e(t)$, where 
    \begin{equation}\label{8-1}
    	\hat{\Gamma}_{t_\iota}=\begin{bmatrix}
    		\sqrt{\rho(\Gamma_{t_\iota})}&\eta_{\text{max}}\\
    		\beta_{t_\iota}(\sqrt{\rho(\Gamma_{t_\iota})}+1)&\hat{\sigma}_{t_\iota}+\eta_{\text{max}}\beta_{t_\iota}
    	\end{bmatrix}.
    \end{equation}
    Therefore, it can be derived from that
    \begin{align}\label{105-2}
    	\Vert e(t_{\iota+1})\Vert_F&\le\Vert \hat{\Gamma}_{t_\iota}^{\kappa_{t_\iota}-1}\Vert_F\Vert e(t_\iota+1)\Vert_F
    \end{align}
    Next, for the case $t=t_\iota$, it can be readily derived that
    \begin{align}\label{105-3}
    	&\Vert Y(t_\iota+1)-\mathbf{1}_{\widetilde{n}_{t_\iota}}x_{t_\iota}^*\Vert_F&\notag\\
    	&\le \sqrt{\rho(\Gamma_{t_\iota})}\Vert Y(t_\iota)-\mathbf{1}_{\widetilde{n}_{t_\iota}}x_{t_\iota}^*\Vert_F+\eta_{\text{max}}\Vert S(t_\iota)-\bar{S}(t_\iota)\Vert_F,
    \end{align}
    \begin{flalign}
    	&\Vert S(t_\iota+1)-\bar{S}(t_\iota+1)\Vert_F&\notag\\
    	&=\Vert \hat{P}_{t_\iota}S(t_\iota)+Q_{t_\iota}(t_\iota+1)-Q_{t_{\iota-1}}(t_\iota)\notag\\
    	&-\hat{\phi}_{t_\iota}\big(S(t_\iota)+Q_{t_\iota}(t_\iota+1)-Q_{t_{\iota-1}}(t_\iota)\big)\Vert_F\notag\\
    	&\le\Vert \hat{P}_{t_\iota} -\hat{\phi}_{t_\iota}\Vert_2\Vert S(t_\iota)-\bar{S}(t_\iota) \Vert_F\notag\\
    	&+\Vert I_{\widetilde{n}_{t_\iota}}-\hat{\phi}_{t_\iota}\Vert_2\Vert Q_{t_\iota}(t_\iota+1)-Q_{t_{\iota-1}}(t_\iota)\Vert_F\notag\\	
    	&\le \hat{\sigma}_{t_\iota}\Vert S(t_\iota)-\bar{S}(t_\iota)\Vert_F+4\beta_{t_\iota} \widetilde{n}_{t_\iota}B_x,\label{25-1}
    \end{flalign} 
    Consequently, it follows from \eqref{105-2},\eqref{105-3} and \eqref{25-1} that
    \begin{align*}
    	\Vert e(t_{\iota+1})\Vert_F&\le\Vert \hat{\Gamma}_{t_\iota}^{\kappa_{t_\iota}-1}\Vert_F\Vert e(t_\iota+1)\Vert_F\notag\\
    	&\le \Vert \hat{\Gamma}_{t_\iota}^{\kappa_{t_\iota}-1}\Vert_F(\Vert e(t_\iota)\Vert_F+4\beta_{t_\iota} \widetilde{n}_{t_\iota}B_x).
    \end{align*}
    
    Finally, by lemma 3 of \cite{olfati2007consensus} and the theorem 8.5.1 of \cite{horn2012matrix}, then it follows that $\sigma_{t_\iota}<1$ and $\hat{\sigma}_{t_\iota}<1$. 
    If the stepsize satisfies \eqref{6-10}, then it follows from corollary 8.1.29 of \cite{horn2012matrix}
    we have $\rho(\hat{\Gamma}_{\iota})<1$.
    Then, we obtain that
    \begin{align*}
    	\lim\limits_{\kappa_{t_\iota}\rightarrow\infty}\Vert e(t_{\iota+1})\Vert_F=0.
    \end{align*}
     This completes the proof.
    
	\end{proof}

\renewcommand{\proofname}{\bf Proof of Theorem 3}
\begin{proof}
	 Denote $a(t)=\Vert \phi_{t_\iota}Y(t)-\mathbf{1}_{\widetilde{n}_{t_\iota}}x_{t_\iota}^*\Vert_F$, $b(t)=\Vert Y_\perp(t)\Vert_F$ and $c(t)=\Vert S(t)-\bar{S}(t)\Vert_F$.
	$E(t)\coloneqq \Vert Y(t)-\mathbf{1}_{\widetilde{n}_{t_\iota}}x_{t_\iota}^* \Vert_F$.
	Next, we analyze the evolution for $t_{\iota}+1\le t\le t_{\iota+1}$.
	From \eqref{6-6}, we have
	\begin{flalign}\label{105-4}
		&\left\Vert P_{t_\iota}Y(t)-\mathbf{1}_{\widetilde{n}_{t_\iota}}x_{t_\iota}^*-\boldsymbol{\eta(t)}\big(\bar{S}(t)-G_{t_\iota}(x_{t_\iota}^*)\big)\right\Vert_F^2&\notag\\
		&\le(1+\eta_{\text{max}}^2(t)\beta_{t_\iota}^2-\frac{2\alpha_{t_\iota}\eta_{\text{min}}(t)}{\sqrt{\widetilde{n}_{t_\iota}}})a^2(t)\notag\\
		&+2\beta_{t_\iota}\eta_{\text{max}}(t)(1+\sigma_{t_\iota})a(t)b(t)\notag\\
		&+(\sigma_{t_\iota}^2+\eta_{\text{max}}^2(t)\beta_{t_\iota}^2+2\beta_{t_\iota}\eta_{\text{max}}(t)\sigma_{t_\iota})b^2(t)\notag\\
		&\le (1+\eta_{\text{max}}^2(t)\beta_{t_\iota}^2-\frac{\alpha_{t_\iota}\eta_{\text{min}}(t)}{\sqrt{\widetilde{n}_{t_\iota}}})a^2(t)\notag\\
		&+\Bigg(\sigma_{t_\iota}^2+\eta_{\text{max}}^2(t)\beta_{t_\iota}^2+2\beta_{t_\iota}\eta_{\text{max}}(t)\sigma_{t_\iota}\notag\\
		&\left.+\frac{\sqrt{\widetilde{n}_{t_\iota}}\eta_{\text{max}}^2(t)\beta_{t_\iota}^2(1+\sigma_{t_\iota})^2}{\alpha_{t_\iota}\eta_{\text{min}}(t)}\right)b^2(t),
	\end{flalign}
	where the last inequality can be obtained by using Young's inequality.
	It is easy to see that $E(t)=a^2(t)+b^2(t)$.
	Furthermore, when $\kappa_{t_\iota}$ is sufficiently large, then there is a positive constant $\gamma_1$, such that $\forall t>\gamma_1$,  $\eta_{\text{max}}^2(t)\beta_{t_\iota}^2+2\beta_{t_\iota}\eta_{\text{max}}(t)\sigma_{t_\iota}+\frac{\sqrt{\widetilde{n}_{t_\iota}}\eta_{\text{max}}^2(t)\beta_{t_\iota}^2(1+\sigma_{t_\iota})^2}{\alpha_{t_\iota}\eta_{\text{min}}(t)}<\frac{1-\sigma_{t_\iota}^2}{2}$ and $\eta_{\text{max}}^2(t)\beta_{t_\iota}^2<\frac{\alpha_{t_\iota}\eta_{\text{min}}(t)}{2\sqrt{\widetilde{n}_{t_\iota}}}$.
	Then, $\forall t>\gamma_1$, it is follows from \eqref{5-3} and \eqref{105-4} that
{\small	\begin{flalign*}
		&E(t+1)&\\
		&\le\sqrt{E^2(t)+\left(\eta_{\text{max}}^2(t)\beta_{t_\iota}^2-\frac{\alpha_{t_\iota}\eta_{\text{min}}(t)}{\sqrt{\widetilde{n}_{t_\iota}}}\right)a^2(t)-\frac{1-\sigma_{t_\iota}^2}{2}b^2(t)}\notag\\
		&+\eta_{\mathrm{max}}(t)c(t)\notag\\
		&\le \sqrt{E^2(t)-\frac{\alpha_{t_\iota}\eta_{\text{min}}(t)}{2\sqrt{\widetilde{n}_{t_\iota}}}a^2(t)-\frac{1-\sigma_{t_\iota}^2}{2}b^2(t)}+\eta_{\mathrm{max}}(t)c(t).
	\end{flalign*}}
	Squaring both sides of the above equation yields
	\begin{flalign*}
		E^2(t+1)\le D^2_0(t)+\eta_{\text{max}}^2(t)c^2(t)+2\eta_{\mathrm{max}}(t)c(t)D_0(t),
	\end{flalign*}
	where $D_0(t)=\sqrt{E^2(t)-\frac{\alpha_{t_\iota}\eta_{\text{min}}(t)}{2\sqrt{\widetilde{n}_{t_\iota}}}a^2(t)-\frac{1-\sigma_{t_\iota}^2}{2}b^2(t)}.$	
	Due to $D_0(t)\le\sqrt{a(t)+b(t)}$, it follows from Young's inequality that
	\begin{flalign*}
		&2\eta_{\text{max}}c(t)D_0(t)\le \frac{\alpha_{t_\iota}\eta_{\text{min}}(t)}{4\sqrt{\widetilde{n}_{t_\iota}}})a^2(t)+D_2 b^2(t)\notag&\\
		&+(\frac{4\eta_{\text{max}}^2(t)\sqrt{\widetilde{n}_{t_\iota}}}{\alpha_{t_\iota}\eta_{\text{min}}(t)}+\frac{4\eta_{\text{max}}^2(t)c^2(t)}{1-\sigma_{t_\iota}^2})c^2(t).
	\end{flalign*} 
	Therefore, we obtain 
	\begin{flalign*}
		&E^2(t+1)\le E^2(t)-D_1 \eta_{\text{min}}(t)a^2(t)-\frac{1-\sigma_{t_\iota}^2}{4}b^2(t)&\notag\\
		&+\left(\frac{4\sqrt{\widetilde{n}_{t_\iota}}\eta_{\mathrm{max}}^2(t)}{\alpha_{t_\iota}\eta_{\text{min}}(t)}+D_3\eta_{\text{max}}^2(t)\right)c^2(t),
	\end{flalign*}
	where $D_1=\frac{\alpha_{t_\iota}}{4\sqrt{\widetilde{n}_{t_\iota}}}$, $D_2=\frac{1-\sigma_{t_\iota}^2}{4}$ and $D_3=1+\frac{4}{1-\sigma_{t_\iota}^2}$.
	
	To further analyze the evolutionary properties of $E(t)$, we next consider the evolutionary properties of $c(t)$.
	First, since $\Omega_{i,j}$ is bounded and 2) of Assumption 1 holds, it is easy to see that $\Vert G_{t_\iota}(x_{t_\iota}^*)\Vert_F$ is bounded.
	Then, there must be a constant $B_G^{t_\iota}$ such that $\Vert G_{t_\iota}(x_{t_\iota}^*)\Vert_F\le B^{t_\iota}_G,$$\forall t_\iota$.
	It follows that
		\begin{flalign*}
		&\Vert S(t)\Vert_F\le\Vert S(t)-\bar{S}(t)\Vert_F+\Vert \bar{S}(t)-G_{t_\iota}(x_{t_\iota}^*)\Vert_F&\notag\\
		&+\Vert G_{t_\iota}(x_{t_\iota}^*)\Vert_F\notag\\
		&\le c(t)+\beta_{t_\iota} E(t)+ B^{t_\iota}_G.
	\end{flalign*}
	Furthermore, we have 
		\begin{flalign*}
		&\Vert Y(t+1)-Y(t)\Vert_F&\notag\\
		&\le\Vert (P_{t_\iota}-I)(Y(t)-\phi_{t_\iota}Y(t))-\mathbf{\eta}(t)S(t)\Vert_F\notag\\
		&\le\Vert P_{t_\iota}-I\Vert_F b(t)+\eta_{\text{max}}(t)\Vert S(t)\Vert_F\notag\\
		&\le \Vert P_{t_\iota}-I\Vert_F b(t)+\eta_{\text{max}}(t)\big(c(t)+\beta_{t_\iota} E(t)+ B^{t_\iota}_G\big).
	\end{flalign*}
	Therefore, it can be obtained by \eqref{6-9}
	\begin{flalign*}
		&c(t+1)\le\hat{\sigma}_{t_\iota}c(t)+\beta_{t_\iota}\Vert Y(t+1)-Y(t)\Vert_F&\notag\\
		&\le D_4(t)c(t)+D_5b(t)+\eta_{\mathrm{max}}(t)\beta_{t_\iota}^2E(t)+D_6\eta_{\mathrm{max}}(t) B_G^{t_\iota},
	\end{flalign*}
	where $D_4(t)=\hat{\sigma}_{t_\iota}+\eta_{\mathrm{max}}\beta_{t_\iota} $, $D_5=\beta_{t_\iota}\Vert P_{t_\iota}-I\Vert_F$ and $D_6=\beta_{t_\iota} B_G^{t_\iota}$.
    Squaring both sides of the above equation yields
    \begin{flalign*}
    	&c^2(t+1)\le D_4^2(t)c^2(t)+D_5^2b^2(t)+\beta_{t_\iota}^4\eta_{\mathrm{max}}^2(t)E^2(t)&\notag\\
    	&+D_6^2\eta_{\mathrm{max}}^2(t)+2D_4(t)D_5b(t)c(t)+2\eta_{\mathrm{max}}(t)\beta_{t_\iota}^2D_4(t)E(t)c(t)\notag\\
    	&+2D_4(t)D_6\eta_{\mathrm{max}}(t)c(t)+2\beta_{t_\iota}^2D_5\eta_{\mathrm{max}}(t)E(t)b(t)\notag\\
    	&+2D_5D_6\eta_{\mathrm{max}}(t)b(t)+2\beta_{t_\iota}^2 D_6\eta_{\mathrm{max}}^2(t)E(t).
    \end{flalign*}
	According to the Young's inequality, there is a positive constant $\epsilon$ such that 
	\begin{flalign*}
		&2D_4(t)D_5b(t)c(t)<\epsilon D_4^2(t)c^2(t)+\frac{1}{\epsilon}D_5^2b^2(t),\\
		&2\eta_{\mathrm{max}}(t)\beta_{t_\iota}^2D_4(t)E(t)c(t)<\epsilon D_4^2(t)c^2(t)+\frac{1}{\epsilon}\beta_{t_\iota}^4\eta_{\mathrm{max}}^2(t)E^2(t),\\
		&2D_4(t)D_6\eta_{\mathrm{max}}(t)c(t)<\epsilon D_4^2(t)c^2(t)+\frac{1}{\epsilon}D_6^2\eta_{\mathrm{max}}^2(t),\\
		&2\beta_{t_\iota}^2D_5\eta_{\mathrm{max}}(t)E(t)b(t)<D_5^2b^2(t)+\beta_{t_\iota}^4\eta_{\mathrm{max}}^2(t)E^2(t),\\
		&2D_5D_6\eta_{\mathrm{max}}(t)b(t)<D_5^2b^2(t)+D_6^2\eta_{\mathrm{max}}^2(t),\\
		&2\beta_{t_\iota}^2 D_6\eta_{\mathrm{max}}^2(t)E(t)<\beta_{t_\iota}^4\eta_{\mathrm{max}}^2(t)E^2(t)+D_6^2\eta_{\mathrm{max}}^2(t).
	\end{flalign*}
   Then, we have 
   \begin{flalign*}
   	&c^2(t+1)\le (1+3\epsilon)D_4^2(t)c^2(t)+\frac{1+3\epsilon}{\epsilon}D_5^2b(t)&\notag\\
   	&+\frac{1+3\epsilon}{\epsilon}\beta_{t_\iota}^4\eta_{\mathrm{max}}^2(t)E^2(t)+\frac{1+3\epsilon}{\epsilon}D_6^2\eta_{\mathrm{max}}^2(t).
   \end{flalign*}
	Obviously, if $t$ is sufficiently large, $D_4(t)<1$ holds.
	In other words, there is positive constant $\gamma_2$ such that for all $t>\gamma_2$, $D_4(t)<1$ holds.
	Let $D_7=(1+3\epsilon)D_4^2(\gamma_2+1)$.
	Since $D_4(t)$ is monotonically decreasing and bounded, it follows from $\epsilon<\frac{1-D_4^2(\gamma_2+1)}{3D_4^2(\gamma_2+1)}$ that $(1+3\epsilon)D_4^2(t)\le D_7<1$ for all $t>\gamma_2$.
	
	Finally, construct $W(t)=E^2(t)+\lambda c^2(t)$, where $\lambda>0$.
	Then, we have 
	\begin{flalign*}
		&W(t+1)-W(t)\le -D_1\eta_{\text{min}}(t)E^2(t)-D_2b^2(t)&\notag\\
		&+D_3\eta_{\mathrm{max}}^2(t)c^2(t)-\lambda c^2(t)+\lambda D_7c^2(t)+\lambda \frac{1+3\epsilon}{\epsilon}D_5^2b^2(t)\notag\\
		&+\lambda\frac{1+3\epsilon}{\epsilon}\beta_{t_\iota}^4\eta_{\mathrm{max}}^2(t)E^2(t)+\lambda \frac{1+3\epsilon}{\epsilon}D_6^2\eta_{\mathrm{max}}^2(t).
	\end{flalign*}
	Obviously, there must be a small enough $\lambda$ such that $\frac{1+3\epsilon}{\epsilon}\lambda D_5^2<\frac{D_2}{2}$.
	Furthermore, there is $\gamma_3>0$ such that for $t>\gamma_3$ 
	\begin{flalign*}
		&\lambda\beta_{t_\iota}^4\frac{1+3\epsilon}{\epsilon}\eta_{\mathrm{max}}^2(t)<\frac{D_1}{4}\eta_{\text{min}}(t),&\\
		&\lambda\eta_{\mathrm{max}}^2(t)<\frac{D_2}{2},\\
		&D_3\eta_{\mathrm{max}}(t)<\frac{\lambda(1-D_7)}{4}.
	\end{flalign*}
	Then, for $t>\max\{\gamma_1,\gamma_2,\gamma_3\}$, 
	\begin{flalign*}
		&W(t+1)-W(t)\le -\frac{D_1\eta_{\text{min}}(t)}{2}E^2(t)-\frac{D_2}{2}b^2(t)&\notag\\
		&-\lambda\frac{1-D_7}{2}c^2(t)+D_8\eta_{\mathrm{max}}^2(t),
	\end{flalign*}
	where $D_8=\lambda \frac{1+3\epsilon}{\epsilon}D_6^2$.

	Let $\hat{W}(t)\coloneq W(t)+D_8\sum\limits_{\tau=t}^\infty\eta_{\mathrm{max}}^2(\tau)$.
	Then, we have 
	\begin{flalign*}
		&\hat{W}(t_{\iota+1})\le \hat{W}(t_{\iota+1}-1)-\frac{D_1\eta_{\text{min}}(t_{\iota+1}-1)}{2}E^2(t_{\iota+1}-1)\notag\\
		&-\frac{D_2}{2}b^2(t_{\iota+1}-1)-\lambda\frac{1-D_7}{2}c^2(t_{\iota+1}-1)\notag\\
		&\le\hat{W}(t_\iota+1)-\frac{D_1}{2}\sum\limits_{\tau=t_\iota+1}^{t_{\iota+1}-1}\eta_{\text{min}}(\tau)E^2(\tau)-\frac{D_2}{2}\sum\limits_{\tau=t_\iota+1}^{t_{\iota+1}-1}b^2(\tau)\notag\\
		&-\lambda\frac{1-D_7}{2}\sum\limits_{\tau=t_\iota+1}^{t_{\iota+1}-1}c^2(\tau).
	\end{flalign*}
	Since $\hat{W}(t)>0$, $\lim\limits_{\kappa_{t_\iota}\rightarrow\infty}\sum\limits_{\tau=t_\iota+1}^{t_{\iota+1}-1}\eta_{\text{min}}(\tau)E^2(\tau)<\infty$, $\lim\limits_{\kappa_{t_\iota}\rightarrow\infty}\sum\limits_{\tau=t_\iota+1}^{t_{\iota+1}-1}b^2(\tau)<\infty$ and $\lim\limits_{\kappa_{t_\iota}\rightarrow\infty}\sum\limits_{\tau=t_\iota+1}^{t_{\iota+1}-1}c^2(\tau)<\infty$.
	Therefore, it follows that $E(t)$ is bounded for $t\in{[t_\iota+1, t_{\iota+1}]}$, $\lim_{t_{\iota+1}\rightarrow\infty}b(t_{\iota+1})=0$ and $\lim_{t_{\iota+1}\rightarrow\infty}c(t_{\iota+1})=0$. 
	Let's assume that the limit of $E(t)$ does not exist, then its upper limit must be greater than zero.
	Without loss of generality, let $\limsup_{t_{\iota+1}\rightarrow\infty}E(t_{\iota+1})=D_9>0.$
	There exists a convergent subsequence $E(t_{\tau})$ limit equal to $D_9$.
	Then, by $\sum_{t}^{\infty}\eta_{\mathrm{min}}(t)=\infty$, it follows that
	$$\infty<\lim\limits_{\kappa_{t_\iota}\rightarrow\infty}\sum\limits_{t_\tau=t_\iota+1}^{t_{\iota+1}-1}\eta_{\text{min}}(t_\tau)E^2(t_{\tau}).$$
	The above formula creates a contradiction with $\lim\limits_{\kappa_{t_\iota}\rightarrow\infty}\sum_{\tau=t_\iota+1}^{t_{\iota+1}-1}\eta_{\text{min}}(\tau)E^2(\tau)<\infty$, so $\lim_{t_{\iota+1}\rightarrow\infty}E(t_{\iota+1})=0$.
	In conclusion, the algorithm with the stepsize \eqref{7-1} converts to the NE when $\kappa_{t_\iota}$ is large enough.
\end{proof}




\bibliographystyle{IEEEtran}
\bibliography{myref4.bib}

\end{document}